\documentclass[twocolumn, prx, nofootinbib]{revtex4}

\usepackage{graphicx}
\usepackage{amsmath}
\usepackage{amsfonts}
\usepackage{amsthm}
\usepackage{enumerate}
\usepackage{xcolor}
\usepackage{stmaryrd}
\usepackage{enumitem}
\usepackage{soul}
\usepackage[normalem]{ulem}

\newcommand{\sigmax}[1]{\sigma^X\left[#1\right]}
\newcommand{\sigmay}[1]{\sigma^Y\left[#1\right]}
\newcommand{\sigmaz}[1]{\sigma^Z\left[#1\right]}

\newcommand{\sigmaX}{\sigma^X}

\newcommand{\Gin}{\mathcal{G}_{\text{in}}}
\newcommand{\Gch}{\mathcal{G}_{\text{ch.}}}
\newcommand{\Gout}{\mathcal{G}_{\text{out}}}
\newcommand{\Gres}{\mathcal{G}_{\mathcal{R}}}

\newcommand{\cent}{\mathcal{C}}
\newcommand{\chan}{\mathcal{K}}

\newcommand{\supp}[1]{\left| #1 \right|}

\theoremstyle{plain}
\newtheorem{thm}{Theorem}
\theoremstyle{plain}
\newtheorem{lem}{Lemma}

\newtheorem*{cor}{Corollary}
\theoremstyle{definition}
\newtheorem{defn}{Definition}

\theoremstyle{remark}

\newenvironment{proof1}{%
	\proof}{\endproof}
\newenvironment{proof1a}{%
	\proof}{\endproof}
	\newenvironment{proof1b}{%
	\proof}{\endproof}
\newenvironment{proof2}{%
	\proof}{\endproof}

\begin{document}
\title{Universal fault-tolerant measurement-based quantum computation}
\author{Benjamin J. Brown}

\author{Sam Roberts}
\affiliation{Centre for Engineered Quantum Systems, School of Physics, University of Sydney, Sydney, New South Wales 2006, Australia}

\begin{abstract}
Certain physical systems that one might consider for fault-tolerant quantum computing where qubits do not readily interact, for instance photons, are better suited for measurement-based quantum-computational protocols. Here we propose a measurement-based model for universal quantum computation that simulates the braiding and fusion of Majorana modes. To derive our model we develop a general framework that maps any scheme of fault-tolerant quantum computation with stabilizer codes into the measurement-based picture. As such, our framework gives an explicit way of producing fault-tolerant models of universal quantum computation with linear optics using protocols developed using the stabilizer formalism. Given the remarkable fault-tolerant properties that Majorana modes promise, the main example we present offers a robust and resource efficient proposal for photonic quantum computation.
\end{abstract}

\maketitle

\section{Introduction}

 Considerable experimental effort~\cite{Walther05, Prevedel07, Tame07, Yao12,barz2014demonstrating, cai2015entanglement, spagnolo2014experimental, carolan2014experimental, carolan2015universal,  silverstone2015qubit, silverstone2016silicon, heuck2016demand, qiang2018large, adcock2018programmable, steinbrecher2018quantum, sparrow2018simulating, harris2018linear, flamini2018photonic} is being dedicated to the realization of linear optical quantum computation~\cite{Knill01, Rudolph16} due to the extensive coherence times that photonic qubits promise. To support this endeavor, we must design robust models of fault-tolerant quantum computation that minimize the high resource cost that is demanded by their implementation. Unlike the other approaches to realize fault-tolerant quantum computation~\cite{Nigg14, Barends14, Kelly15, Corcoles15, Takita16,  Linke16, Albrecht16, Deng16}, qubits encoded with photons fly through space at the speed of light, and do not readily interact with other photons. As such, we need different theoretical models to describe fault-tolerant quantum computation with photons.

Measurement-based quantum computation~\cite{Briegel01, Raussendorf01, Raussendorf02, Briegel09,nielsen2005fault} provides a natural language to describe computational operations with flying qubits. In this picture we initialize a specific many-body entangled resource state which is commonly known as a cluster state. We then make single-qubit measurements on its qubits to realize computational operations. The performed operation is determined by the choice of measurements we make, and the entangled resource we initially produced.

In the present work, we provide a fault-tolerant implementation of the Clifford operations with an adaptation of the topological cluster-state model~\cite{Raussendorf05}. If supplemented by a non-Clifford gate~\cite{Brown19a}, that one might implement by magic state distillation~\cite{Bravyi05, Campbell16}, these operations are enough to recover universal quantum computation~\cite{Raussendorf07, Raussendorf07a}.
Owing to its high threshold error rates~\cite{Raussendorf07, Raussendorf07a, Fowler09} and simplicity in its design~\cite{Devitt09, Herrera-Marti10, Gimeno-Sergovia15,  Rudolph16}, the topological cluster state is the prototypical model for robust fault-tolerant measurement-based quantum computation~\cite{Raussendorf06}. The number of two-qubit entangling operations we require per physical qubit to realize our gate set are no more demanding than those to realize the standard topological cluster state. We therefore expect that the Clifford gates we propose will maintain the high threshold of the topological cluster state. As such, our model is particularly appealing for experimental realization.

The potential savings available by adopting our robust model of quantum computation are twofold. Firstly, we achieve the Clifford operations by simulating the braids and fusion measurements of Majorana zero modes with the topological cluster-state model. We find these operations by mapping them from analogous operations using stabilizer models that are known to be resource efficient~\cite{Bombin10, Hastings15, Brown17, Yoder17} compared with the original proposal~\cite{Raussendorf06}. Furthermore, unlike our model, the original proposal relied on resource-intensive distillation methods to complete the set of Clifford operations~\cite{Jones12}. We expect additional reductions in resource demands by circumventing the need to perform any distillation operations to perform Clifford gates.

We derive our results by developing a framework that allows us to map any protocol for quantum computation with stabilizer codes~\cite{Gottesman97} onto a measurement-based model. Our mapping is therefore important in its own right as it enables us to import all of the developments made for fault-tolerant quantum computation with static qubits using the stabilizer formalism in the past decades~\cite{Terhal15, Brown16, Campbell17} into a language that is better suited for linear-optical setups. Specifically, we show how to encode an arbitrary stabilizer state with a cluster-state model. We also show how to choose the cluster state and the measurement pattern such that we implement logical operations on the encoded stabilizer code. We accomplish logical operations using a specially chosen quantum measurements to perform code deformations~\cite{Raussendorf06, Bombin09, Paetznick13, Bombin15, Anderson14, Brown17, Vuillot18}. We show that we can encode code deformations within a resource state such that after we measure most of its qubits, the remaining unmeasured qubits of the resource state lie in the state of the deformed stabilizer code, thereby completing a fault-tolerant operation. We also show how we can compose many of these operations, that we will often call channels, to realize longer computations and more complex code deformations.

Since the work of Raussendorf~{\it et al.} it has been shown~\cite{Bolt16} how to map a special class of stabilizer codes, known as Calderbank-Shor Steane(CSS) stabilizer codes~\cite{Calderbank96, Steane96}, onto a measurement-based model. We have extended beyond this work by showing how to perform computations via code deformations on  arbitrary stabilizer codes from the measurement-based perspective. Further, in Ref.~\cite{Herr17b} it was proposed that lattice surgery~\cite{Horsman12, Landahl14} can be mapped onto a measurement-based computational model. As we will discuss, the specific model proposed in Ref.~\cite{Herr17b}, together with any other model of code deformation~\cite{Raussendorf06, Raussendorf07, Raussendorf07a}, can be reproduced using our framework. We finally remark on recent work~\cite{Nickerson18, Newman19} on measurement-based quantum computation where the topological cluster-state model is generalized to find robust codes, and Refs.~\cite{Bombin18, Brown19a} where a new schemes for universal fault-tolerant measurement-based quantum computation are proposed based using three-dimensional codes~\cite{Bombin07, Vasmer18}.

The remainder of this article is organised as follows. We begin by reviewing notation we use to describe quantum error-correcting codes in Sec.~\ref{Sec:StabilizerCodes}. After introducing some basic notation and the concept of foliation, we summarize the results of our paper and give a guide to the reader to parse the different aspects of our model in Sec.~\ref{Sec:FoliationModel}. Then, in Sec.~\ref{Sec:FoliatedQubit} develop a microscopic model for the one-dimensional cluster state as a simple instance of foliation, and we consider parity measurements between separate foliated qubits. In Sec.~\ref{Sec:CodeFoliation} we use the microscopic framework we build to show how a channel system can propagate an input stabilizer code unchanged. We explicitly demonstrate this by foliating the twisted surface code model in Sec.~\ref{Sec:TwistedSurfaceCode}. In Sec.~\ref{Sec:FTQC} we go on to show that we can manipulate input states with a careful choice of channel systems. We demonstrate this by showing we can perform Clifford gates and prepare noisy magic states with the foliated surface code before offering some concluding remarks. Appendices describe extensions of our results and proofs of technical details. Appendix~\ref{App:TypeII} describes an alternative type of foliated qubit, Appendix~\ref{App:ThmProofs} gives proofs of technical theorems we state in the main text, Appendix~\ref{App:SubsystemFoliation} discusses foliated subsystem codes and Appendix~\ref{App:Compressed} gives a generalization of the foliated stabilizer codes we propose.

\section{Quantum error correction}
\label{Sec:StabilizerCodes}

Here we introduce the notion of a subsystem code~\cite{Poulin05} that we use to describe the foliated systems of interest. A subsystem code is a generalization of a stabilizer code~\cite{Calderbank97, Gottesman97} where not all of the logical operators of the code are used to encode logical information. The logical operators for each of the disregarded logical qubits are known as gauge operators. Gauge operators have been shown to be useful for a number of other purposes~\cite{Bacon06, Paetznick13, Terhal15, Brown16, Campbell17}.

A subsystem code is specified by its gauge group, $\mathcal{G} \subseteq \mathcal{P}_n $; a subgroup of the Pauli group acting on $n$ qubits. The Pauli group is generated by Pauli operators $X_j$ and $Z_j$, together with the phase $\text{i}$, where the index $1 \le j \le n$ denotes the code qubit the operator acts on. More precisely we have
\begin{equation}
P_j = \underbrace{\openone \otimes \openone \otimes \dots \otimes \openone}_{j-1}\otimes P \otimes \underbrace{\openone \otimes \dots \otimes \openone}_{n - j},
\end{equation}
where $\openone$ is the two-by-two identity matrix and $P = X, Y, Z$ is a Pauli matrix.

The gauge group describes a code space specified by its stabilizer group 
\begin{equation}
\mathcal{S} \propto \cent(\mathcal{G}) \cap \mathcal{G},
\end{equation} 
where $\mathcal{C}(\mathcal{G})$ denotes the centralizer of a group $\mathcal{G}$ within $\mathcal{P}_n$ which consists of all elements of $\mathcal{P}_n$ that commute with all elements of $\mathcal{G}$. With the stabilizer group defined we specify the code space as the subspace spanned by a basis of state vectors $|\psi \rangle$ where
\begin{equation}
s | \psi \rangle = (+1) | \psi \rangle,
\end{equation}
for all stabilizers $s \in \mathcal{S}$. By definition, the stabilizer group must satisfy $-1  \not\in \mathcal{S}$.

We also consider a generating set of logical operators $\mathcal{L} = \cent(\mathcal{G}) \backslash \mathcal{G}$. The group $\mathcal{L}$ is generated by the logical Pauli operators $\overline{X}_j$, $\overline{Z}_j$ with $1 \le j \le k$, such that $\overline{X}_j$ anti-commutes with $\overline{Z}_l$ if and only if $j = l$. Otherwise all generators of the group of logical Pauli operators commute with one another.

The logical operators generate rotations within the code space of the stabilizer code. We will frequently make use of the fact that logical operators $L,\,L' \in \mathcal{L} $ such that $L' = s L$ with $s\in \mathcal{S}$ have an equivalent action on the code space. This follows from the definitions given above. We thus use the symbol `$\cdot \sim \cdot $' to denote that two operators are equivalent up to multiplication by a stabilizer operator. For instance, with the given example we can write $L' \sim L$.

It is finally worth noting that the special Abelian subclass of subsystem codes, namely stabilizer codes~\cite{Calderbank97, Gottesman97}, are such that $\mathcal{S} = \mathcal{G}$ up to phases.

\subsection{Transformations and compositions of codes}

It will be important to make unitary maps between subsystem codes. Given the generating set of two different stabilizer codes $\mathcal{R}$ and $\mathcal{S}$ with elements $r \in \mathcal{R}$ and $s\in \mathcal{S}$, the stabilizer group $\mathcal{T} = \mathcal{R} \otimes \mathcal{S}$ is generated by elements $r \otimes 1 ,\, 1 \otimes s \in \mathcal{T}$. We also use the shorthand $\mathcal{T} = U \mathcal{S}$ to define the stabilizer group
\begin{equation}
\mathcal{T} = \left\{ U s U^\dagger ~:~ s \in \mathcal{S}  \right\},
\end{equation}
where $U$ is a Clifford operator. Of course, the commutation relations between two Pauli operators are invariant under conjugation by a unitary operator.

\subsection{Expressing Pauli operators as vectors}

For situations where one is willing to neglect the phases of elements of the stabilizer group, it is common to write elements of the Pauli group as vectors of a $2n$-dimensional vector space over a binary field with a symplectic form that captures their commutation relations~\cite{Calderbank97, Haah12}.
We express Pauli operators in vector notation such that $p = (p^X \, p^Z )^T$ where $p^X$($p^Z$) are vectors from the $n$-dimensional vector space over the binary field $\mathbb{Z}_2$ and the superscript $T$ denotes the transpose of the vector such that, up to phases, the Pauli operator $P \in \mathcal{P}_n$ is expressed $P = \prod_{j=1}^n X_j^{p^X_j} Z_j^{p^Z_j}$.

We will frequently move between Pauli operators and the vector notation. It is thus helpful to define the function $v : \mathcal{P}_n \rightarrow \mathbb{Z}_2^{2n}$ such that
\begin{equation}
v(P) \equiv \left( p^X \, p^Z \right)^T, \label{Eqn:Vectorise}
\end{equation}
where vectors $p^X,\, p^Z \in \mathbb{Z}_2^n$ are such that
\begin{equation}
P \propto \prod_{j \in \supp{p^X}}X_j \prod_{j \in \supp{p^Z}}Z_j,
\end{equation}
up to a phase factor and we have defined the support of vector $p$, denoted $\supp{p}$, as the set of elements of $ p $ that are non-zero.

Using this notation, we additionally have the symplectic form where, for two vectors $p$ and $q$ specifying two elements of the Pauli group $\mathcal{P}_n$, we have
\begin{equation}
\Upsilon(p,q) \equiv p^T \lambda_n q, \text{ with } \lambda_n = \left( \begin{array}{cc} 0 & \openone_n \\ \openone_n & 0 \end{array} \right), \label{Eqn:Upsilon}
\end{equation}
where addition is taken modulo 2 and $\openone_n$ is the $n \times n$ identity matrix such that $\Upsilon(p,q) = 0$ if and only if the Pauli operators specified by $p$ and $q$ commute. We also define the inner product
\begin{equation}
p \cdot q \equiv \sum_j p_j q_j,
\end{equation}
where addition is taken modulo 2.

\section{Foliation}
\label{Sec:FoliationModel}

Quantum computation proceeds by using a series of channels where each channel maps its input nontrivially onto some output state~\cite{Raussendorf03}. These channels are commonly known as gates, and with an appropriate composition of said channels we can realize non-trivial quantum algorithms. In this work we build a generic model that takes an arbitrary fault-tolerant channel that is based on stabilizer codes, and provides a measurement-based protocol that performs the same function.

Channels are readily composed by unifying the output of the last with the input of the next, so here we focus on developing a model of a single channel with a single designated function. Most important is that the channel is tolerant to errors, and as such we show that we can propagate quantum error-correcting codes through a channel in a fault-tolerant manner. The process of constructing such channels -- which are comprised of a resource state and measurement pattern -- is known as foliation. Here, the term foliation refers to the layered like structure of the resource state that is constructed for each channel which arises from an underlying error correcting code. While a foliated system is relatively straight forward to understand in comparison to its analogous circuit-based counterpart, its microscopic details can become quite obtuse without subdividing the system into constituent parts that depend on their function. 

In what follows is a macroscopic overview for the model of foliation we consider with some description of the function of each part. We conclude our overview with a reader's guide which outlines which subdivisions of the total system each section addresses. Nonetheless, the reader should bear the macroscopic structure presented in this section in mind throughout our exposition.

\subsection{The model}

We look to build a foliated system, denoted $\mathcal{F}$. The channel consists of two components; a resource state, $\mathcal{R}$, and a measurement pattern, $\mathcal{M}$, that propagates the input state that is encoded within $\mathcal{R}$ onto the output system. The union of both $\mathcal{M}$ and $\mathcal{R}$ can be regarded as a generating set of a subsystem code. 

The resource state is a specially-prepared many-body entangled state known as a graph-state that we describe in more detail shortly. The measurement pattern is a list of single-qubit measurements that are performed on the physical qubits not included in the output system. The measurement pattern is chosen specifically to move the input state through the resource onto the output system up to the data collected from the single-qubit measurements. In addition to this, the measurement data obtained is used to determine the physical qubits that have experienced errors during the preparation of the resource state or during the readout process.

Abstractly, we can regard the foliation as a gauge-fixing procedure~\cite{Paetznick13, Anderson14, Bombin15, Vuillot18} where the foliated system is described by the gauge group
\begin{equation}
\mathcal{F} = \mathcal{R} \cup \mathcal{M}, \label{Eqn:Foliated}
\end{equation}
where we use the symbol `$\cdot \cup \cdot$' to denote the group generated by elements of both $\mathcal{R}$ and $\mathcal{M}$. Specifically, the union symbol is the free product between two groups. Foliation then is the procedure of preparing the system $\mathcal{R}$ and subsequently fixing the gauge of $\mathcal{F}$ by measuring $\mathcal{M}$. The preparation of $\mathcal{R}$ and choice of $\mathcal{M}$ determines the action of the channel, and the data we obtain to identify the locations of errors.

Advantageous to the model we present is that, provided we have a channel that is consistent with some chosen stabilizer code in a sense we make precise shortly, we need only know the input code to determine the logical effect of the channel on the output system. The microscopic model can be used to test and compare the tolerance that different foliated systems have to physical faults.

We next briefly elaborate on the resource state $\mathcal{R}$. The resource state can be decomposed into an entangled channel system $\chan$ and ancilla qubits $\mathcal{A}$ together with a unitary operator $U^A$ which couples the ancilla qubits to the channel
\begin{equation}
\mathcal{R} = U^A \left( \chan \otimes \mathcal{A} \right). \label{Eqn:Resource}
\end{equation}
The unitary $U^A$ takes the form of a circuit of controlled-$Z$ operations. Supposing that $\mathcal{R} = \chan$, the resource state will propagate the input system onto the output system provided no errors have occurred before the measurements have been performed and interpreted. The channel system accomplishes this function using a series of one-dimensional cluster states that propagate the input system onto its output via single-qubit measurements.

The resource state also includes an ancilla system. Once coupled to the channel, the ancilla qubits perform two roles upon measurement. First of all they are included to perform check measurements on the channel system to identify errors that may be introduced to the physical qubits. Their second role is to modify the input system through the channel as it progresses to the output system. This modification is analogous to code deformation in the more familiar circuit-based model of quantum computation~\cite{Bombin09, Fowler09, Fowler12a, Brown17}. Later we will investigate how different choices of channel affect the transformation made on the input state.

\subsection{A guide for the reader}

The following exposition follows a number of avenues to describe both the details of a foliated systems at the microscopic level as well as their function at the macroscopic level, but ultimately all the sections have the model presented above in common. We thus provide a guide to help explain the aspects of the fault-tolerant measurement-based model we build over the course of this article.

\begin{figure}
\includegraphics{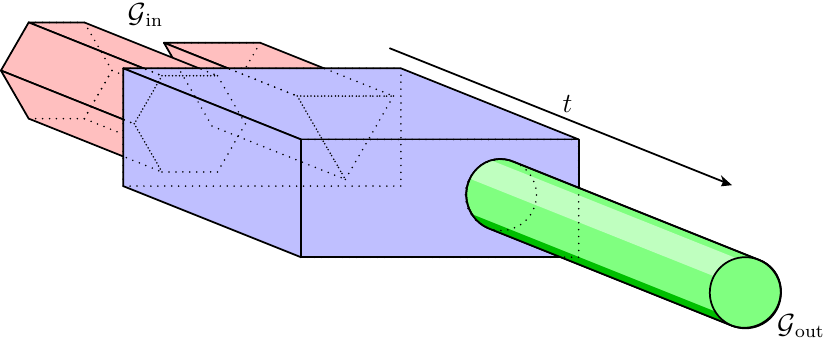}
\caption{\label{Fig:Composition} Three resource states shown in red, blue and green are composed to map an input state onto the output state under a mapping that is determined by the different resource states. A temporal direction can be assigned such that the input system is input at the initial time and the output system emerges at the final instance.}
\end{figure}

The main results are stated most abstractly in Sec.~\ref{Sec:CodeFoliation}. Here we describe the action a foliated system will have on an input code $\mathcal{G}_\text{in}$ which is manifest at the output system $\mathcal{G}_\text{out}$ once the data from the single-qubit measurement pattern has been collected and processed. The section also describes the check operators $\mathcal{G}_\mathcal{R}$ that are collected at the microscopic level of the foliated system that are used to identify the locations of errors. The measurements $\mathcal{G}_{\mathcal{R}}$ generate the stabilizer code $\mathcal{G}_{\text{ch.}}$. These measurements are important for the the error-correction procedure for the foliated system. The choice of $\mathcal{G}_{\mathcal{R}}$ also determines the action the foliated system has on the input code. We summarize the notation we use to describe different components of the foliated system as we describe it in Table.~\ref{Tab:Notation}.

\begin{table}
\begin{tabular}{c  p{2.84in}}
Symbol & Definition \\
\hline \hline
$\mathcal{F}$ & Gauge group of the foliated system $\mathcal{F} = \mathcal{R} \cup \mathcal{M}$. \\
$\mathcal{M}$ & Single-qubit Pauli measurement pattern. \\
$\mathcal{R}$ & Resource state whereby $\mathcal{R} = U^A \left( \mathcal{K} \otimes \mathcal{A} \right) $. \\
$\mathcal{K}$ & Series of linear cluster states that transmit $\mathcal{G}_\text{in}$.\\
$\mathcal{A}$ & Ancilla qubits that are used to identify errors in $\mathcal{R}$. \\
$U^A$ & Unitary that couples $\mathcal{K}$ and $\mathcal{A}$ to give $\mathcal{R}$. \\
$\mathcal{G}_\text{in}$ & Input stabilizer code into the resource state via $\mathcal{K}$.  \\
$\mathcal{G}_\text{out}$ & Output code of foliation obtained by measuring $\mathcal{M}$. \\
$\mathcal{G}_\mathcal{R}$ & Measurements the foliated system performs on $\mathcal{G}_{\text{in}}$.  \\
$\mathcal{G}_\text{ch.}$ & Stabilizer group generated by measurements $\mathcal{G}_{\mathcal{R}}$. \\ 
\end{tabular}
\caption{\label{Tab:Notation} A summary of the notation we use to describe the different components of the foliated system.}
\end{table}

In order to understand the details of Sec.~\ref{Sec:CodeFoliation} we must first examine closely the one-dimensional cluster state that propagates qubits through an entangled system of physical qubits. In Sec.~\ref{Sec:FoliatedQubit} we study the one-dimensional system and we build notation to describe parity measurements between qubits propagated through a series of one-dimensional cluster-states. These one-dimensional cluster states, and the parity measurements we conduct between them, make up the measurement-based channel for a quantum error-correcting code. The details presented in this section are necessary to understand the technical aspects of the theorems we give in Sec.~\ref{Sec:CodeFoliation}.

Beyond Sec.~\ref{Sec:CodeFoliation} we look at specific instances of the general theory we develop. Indeed, in Sec.~\ref{Sec:TwistedSurfaceCode} we examine the microscopic details of the foliated variant of the twisted surface code. In Sec.~\ref{Sec:FTQC} we show how we can compose different channels to realize fault-tolerant measurement-based quantum computation using the example of surface code quantum computation.

In Fig.~\ref{Fig:Composition} we show a schematic diagram of how fault-tolerant measurement-based quantum computation will proceed. The figure shows the composition of three channels, each of which will perform a different operation. The system is measured such that the logical data is mapped from the input system onto the output system, under an operation that is determined by the choice of different resource states.

\section{Foliated qubits}
\label{Sec:FoliatedQubit}

Here we look at a one-dimensional cluster state which propagates a single qubit along a chain, or `wire', of locally entangled qubits via single-qubit measurements. We also show how to perform parity measurements between several chains using additional ancilla qubits. Both of these components are essential to the fault-tolerant models we will propose in later sections.

\subsection{The one-dimensional cluster state}
\label{Subsec:OneDcluster}

The cluster state is readily described using the stabilizer formalism introduced above. We will denote Pauli matrices acting on the physical qubits, indexed $\mu$, with operators $\sigmax{\mu}$, $\sigmay{\mu}$ and $\sigmaz{\mu}$. This will discriminate the Pauli matrices that act on the physical qubits from those that act on the code qubits of the input and output quantum error-correcting code. Similarly, the logical operators of the cluster state are denoted $X$, $Y$ and $Z$ without the bar notation for the same reason. The logical qubits of the one-dimensional cluster-state wire will become the code qubits of foliated stabilizer codes as this discussion progresses. To this end, wherever there is ambiguity, we will refer to the qubits that lie in a cluster state as `physical qubits'. These qubits are not to be confused with the `code qubits' of a quantum error-correcting code.

To describe the cluster state we first consider the initial product state $ |\psi \rangle _1 | + \rangle _2 | + \rangle_3 \dots | + \rangle _N $ where $ | \psi \rangle $ is an arbitrary single-qubit state, the states $| \pm \rangle$ are eigenstates of the Pauli-X matrix, i.e., $ \sigmaX | \pm \rangle = (\pm 1) | \pm \rangle$. We can express this state with the stabilizer group $\mathcal{I} = \left \langle \{ \sigmax{\mu} \}_{\mu = 2}^N \right \rangle$ whose logical operators are $X = \sigmax{1}$ and $Z = \sigmaz{1}$. The cluster state, whose stabilizer group we denote as $\chan = U \mathcal{I}$, with $U$ defined as
\begin{equation}
U = \prod_{\mu=1}^{N-1} U^Z[\mu,\mu+1], \label{Eqn:ChainEntanglingOperation}
\end{equation}
where $U^Z[\mu,\nu] = (1 + \sigmaz{\mu} + \sigmaz{\nu} - \sigmaz{\mu} \sigmaz{\nu})/2$ is the controlled-phase gate. The unitary operator $U$ couples nearest-neighbour pairs of physical qubits along the open chain. 

The following facts about the controlled-phase gate are helpful throughout our exposition. We firstly note that operators $U^Z[\mu, \nu]$ and $\sigmaz{\rho}$ commute for an arbitrary choice of $\mu, \,\nu$ and $\rho$. Moreover, $U^Z[\mu, \nu]$ and $U^Z[\rho, \lambda]$ commute for any $\mu$, $\nu$, $\rho$ and $\lambda$. Further, they satisfy the relationship
\begin{equation}
U^Z[\mu, \nu] \sigmax{\mu} U^Z[\mu, \nu] = \sigmax{\mu} \sigmaz{\nu}. \label{Eqn:CZPauliX}
\end{equation}
We also have that $U^Z[\mu ,\nu] = U^Z[\nu, \mu]$ by definition.

With the above definitions it is readily checked that $\chan$ is generated by stabilizer operators
\begin{equation}
C[\mu] = \sigmaz{\mu-1} \sigmax{\mu} \sigmaz{\mu+1}, \label{Eqn:ClusterStab1}
\end{equation}
for $ 2 \le \mu \le N-1$, along with 
\begin{equation}
C[N] = \sigmaz{N-1} \sigmax{N}.  \label{Eqn:ClusterStab2}
\end{equation} 
The logical operators that act on the encoded qubit are
\begin{equation}
X = \sigmax{1} \sigmaz{2}, \quad Z = \sigmaz{1}. \label{Eqn:ClusterLogicals}
\end{equation}
We show examples of a stabilizer operator, and the logical Pauli-X and Pauli-Z operator in Fig.~\ref{Fig:OneDcluster}. As is common when describing cluster states, we use a graphical notation to describe the resource state of interest. In particular, pairs of qubits that are coupled via a controlled-phase gate are connected by an edge of a graph where each qubit is represented by vertex.

\begin{figure}
\includegraphics{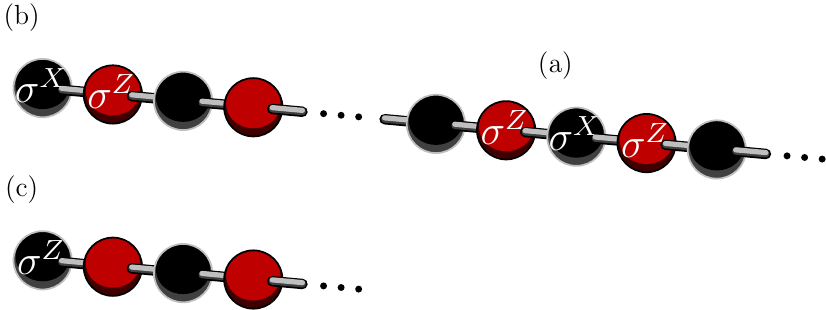}
\caption{\label{Fig:OneDcluster} The one-dimensional cluster state. Physical qubits are shown as vertices ordered along a line where the left-most vertex represents the first qubit of the system and edges indicated pairs of vertices that are entangled via a controlled-phase gate. (a)~A stabilizer operator denoted $C[\mu]$ shown at an arbitrary point along the lattice. (b)~and (c)~show, respectively, representatives of the logical Pauli-X and Pauli-Z operators once the cluster state is initialized.}
\end{figure}

The goal of the one-dimensional cluster state is to transport the logical information along the chain onto the last qubit, indexed $N$. To do so we make single qubit measurements on all of the qubits except qubit $N$. We study two different types of foliated qubits. The first, which is already well understood in the literature, transmits information by measuring the physical qubits in the Pauli-X basis. The second, which does not appear in the literature to the best of our knowledge, moves information using Pauli-Y measurements. We refer to these two foliated qubits as type-I and type-II foliated qubits respectively. While we find that using type-I qubits is sufficient to realize any foliated system of interest within the scope we set here, we believe that there may be practical advantages to be gleaned using type-II qubits in a foliated scheme. As such we describe type-II foliated qubits in App.~\ref{App:TypeII} and we discuss their potential applications throughout our exposition. For simplicity though we will by and large focus on type-I foliated qubits in the main text.

\subsection{Measurement-based qubit transmission}

We review here a foliated qubit where the physical qubits of the entangled chain are measured in the Pauli-X basis. The stabilizer group of the chain $\chan$ is defined above. We first look at the action of measuring the first physical qubit of the system. In particular we are interested in the action of the measurement on the logical operators. This action is easily understood by finding logical operators that commute with the measurement. We find logical operators that commute with the measurement operator $M_1 = \sigmax{1}$ by multiplying the logical operators by stabilizer operators.  We find that the logical operator
\begin{equation}
Z \sim \sigmax{2} \sigmaz{3},
\end{equation}
commutes with $\sigmax{1}$. Similarly $X$, as defined in Eqn.~(\ref{Eqn:ClusterLogicals}), commutes with the measurement operator $M_1$. After making the measurement we project the code such that we have a new stabilizer $C[1] \in \chan$ where $C[1] \equiv x_1M_1 = x_1\sigmax{1}$, where $x_1$ is the measurement outcome of $M_1$. Multiplying $X$ by $C_1$ we have
\begin{equation}
X \sim x_1 \sigmaz{2}. 
\end{equation} 
The stabilizer $C[2]$ is removed from $\chan$ as it anti-commutes with $M_1$. Having accounted for the measurement outcome we can also disregard the first qubit of the chain since the projective measurement has disentangled it from the rest of the system. 

\begin{figure}
\includegraphics{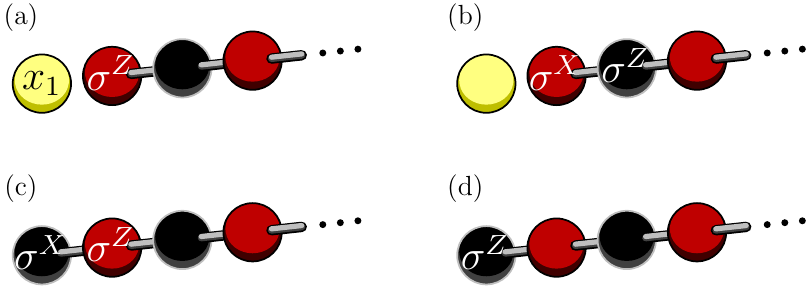}
\caption{\label{Fig:TypeIIlogical} The logical operators of the one-dimensional cluster state after and before the first qubit is measured in the Pauli-X basis. In (a)~we show the logical Pauli-X operator and in (b)~we show the logical Pauli-Z operator after the first qubit is measured, and the measurement outcome $x_1$ is returned. Importantly, after the measurement the logical Pauli-X operator is a single-qubit Pauli-Z operator on the second qubit, whereas the logical Pauli-Z operator is supported on two physical qubits. In contrast, before the measurement the logical Pauli-Z operator is supported on a single qubit, where as the logical Pauli-X operator is a weight-two operator. We show the logical Pauli-X and logical Pauli-Z operator before the measurement in (c)~and (d),~respectively for comparison. These are the same operators as those shown in Fig.~\ref{Fig:OneDcluster}(b) and~(c), respectively.}
\end{figure}

It is important to note that the logical operator $X$, up to the sign determined by the measurement outcome, is now a single-qubit Pauli-Z operator acting on the second qubit. In contrast, before the measurement, the logical operator $Z $ that was a single-qubit Pauli-Z operator acting on the first qubit has now become a weight-two operator acting on the second and third physical qubits along the chain. We show the logical Pauli-X and Pauli-Z operators after the first qubit has been measured in Figs.~\ref{Fig:TypeIIlogical}(a) and (b). We compare these logical operators with the same logical operators before the measurement has taken place in Figs.~\ref{Fig:TypeIIlogical}(c) and~(d), respectively.

The advantage of measurement-based quantum computation lies in the fact that once the resource state is prepared, quantum information can be transmitted and processed by performing single-qubit measurements on qubits of the resource state. Supposing then that we can only make single-qubit measurements, we see that we can measure the logical operator $Z$ with the single-qubit measurement $\sigmaz{1}$ on the first qubit. Conversely, to infer the value of $X$, we measure the first qubit in the Pauli-X basis, and the second qubit in the computational basis. We thus see that we can learn either the logical Pauli-X or Pauli-Z information from the cluster state by measuring the appropriate physical qubit in the Pauli-Z basis, and the qubits that preceded it in the Pauli-X basis. Alternatively, were we to measure both qubit 1 and qubit 2 in the Pauli-X basis, we would have moved the logical information along the chain without having inferred any logical data.

To find a general expression for the logical operators, suppose we have measured the first $\tau-1$ physical qubits in the Pauli-X basis which returned outcomes $x_\mu = \pm 1$ for $\mu  < \tau$. We then have that
\begin{equation}
X \sim \left( \prod_{\mu = 1}^{\tau/2} x_{2\mu -1} \right) \sigmaz{\tau} , 
\end{equation} 
\begin{equation}
Z \sim \left( \prod_{\mu = 1}^{\tau/2} x_{2\mu} \right) \sigmax{\tau} \sigmaz{\tau+1},
\end{equation}
for even $\tau$ and
\begin{equation}
X \sim \left( \prod_{\mu = 1}^{(\tau-1)/2} x_{2\mu -1} \right) \sigmax{\tau} \sigmaz{\tau+1} ,
\end{equation} 
\begin{equation}
Z \sim  \left( \prod_{\mu = 1}^{(\tau-1)/2} x_{2\mu} \right) \sigmaz{\tau},
\end{equation}
where $\tau$ is odd. With this, we see that we can learn the logical Pauli-X information by making a single-qubit Pauli-Z measurement on a site where $\tau$ is even given that we have the outcomes of physical Pauli-X measurements made on the first $\tau-1$ qubits. Likewise we can learn the logical Pauli-Z information by making a physical Pauli-Z measurement on a site where where $\tau$ is odd, provided we have measured all of the qubits with $\mu < \tau$ in the Pauli-X basis. To this end, we find that it is particularly convenient to use a new notation to index the qubits of the chain. We define
\begin{equation}
X(t) = 2t, \quad Z(t) = 2t-1, \label{Eqn:Indices}
\end{equation}
such that now we can rewrite the logical operators of the system such that
\begin{equation}
X \sim \Sigma^X(t) \sigmaz{X(t)}, \quad Z \sim  \Sigma^Z(t) \sigmaz{Z(t)},
\end{equation}
for any $t$, where the operators
\begin{equation}
\Sigma^X(t) \equiv \prod_{\mu = 1} ^t \sigmax{Z(\mu)},\quad \Sigma^Z(t) \equiv \prod_{\mu = 1} ^{t-1} \sigmax{X(\mu)}. \label{Eqn:SigmaTypeI}
\end{equation}
Importantly, the operators $\Sigma^X(t)$ and $\Sigma^Z(t)$ are the tensor product of the Pauli-X matrix. This means the eigenvalues of these operators are inferred from the single-qubit measurement pattern if all the qubits along the chain are measured in the Pauli-X basis as is the case for type-I foliated qubits. 
 
The above redefinition of indices is such that at each `time' interval, indexed by $t$, we can recover either the logical Pauli-X or Pauli-Z information from the chain with a single-qubit Pauli-Z measurement provided the previous qubits along the chain have been measured in the Pauli-X basis. Specifically, each interval contains two adjacent qubits of the chain, the first, which lies at an odd site $ \tau = 2t - 1$, gives access to the Pauli-Z information via a single-qubit measurement, and at every second site, $\tau= 2t$, we can learn logical Pauli-X information with a single-qubit Pauli-Z measurement. We show the logical operators at a given time interval in Fig.~\ref{Fig:TypeIinterval}.

\begin{figure}
\includegraphics{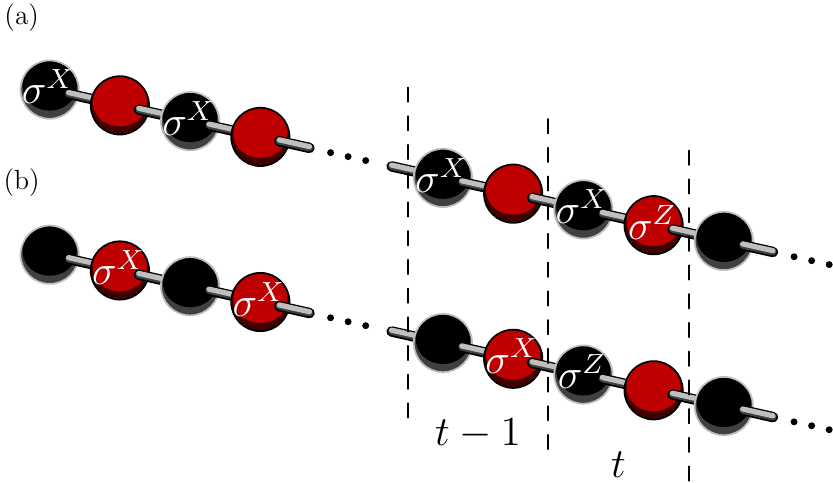}
\caption{\label{Fig:TypeIinterval} (a)~Logical Pauli-X and (b)~logical Pauli-Z operators at time interval $t$. Each time interval contains one black vertex and one red vertex. Local to each time interval we can adjust our single-qubit measurement pattern by measuring the appropriate qubit in the Pauli-Z operator to measure either the logical Pauli-X or Pauli-Z operator given that the preceeding qubits of the system are measured in the Pauli-X basis.}
\end{figure}

\subsection{Measurements using ancilla}

We have thus far imagined replacing the single-qubit Pauli-X measurement on some appropriately chosen qubit with a Pauli-Z measurement to perform logical measurements on the propagated information. To develop foliated codes further we will require the ability to perform logical measurements without adapting the measurement pattern of the foliated qubits. Instead we couple extra ancilla qubits to the system to learn logical information at a given time interval.

To show how to make a logical measurement of a one-dimensional cluster state with an ancilla we consider the resource state $\mathcal{R} = U^A \left( \chan \otimes \mathcal{A} \right)$, as we have introduced in Eqn.~(\ref{Eqn:Resource}), where the channel $\chan$ is the one-dimensional cluster state described above and $\mathcal{A} = \left\langle \sigmax{a} \right\rangle$ describes the stabilizer group of a single ancilla qubit prepared in the $|+\rangle_a$ state. We couple the ancilla to $\chan$ with unitary $U^A$ which we specify shortly. We use the elements of $\mathcal{A}$ to measure logical information from $\chan$.

We must specify a measurement pattern to carry out the propagation of information, as well as the single-qubit measurement we make on the ancilla to learn logical information from the resource state. We write the pattern of measurements 
\begin{equation}
\mathcal{M} =\mathcal{M}^C \cup \mathcal{M}^A,
\end{equation} 
where $\mathcal{M}^C$($\mathcal{M}^A$) describes the measurements made on the qubits of subsystem $\chan$($\mathcal{A}$). For the case of type-I foliated qubits discussed previously we have $\mathcal{M}^C = \left\langle \left\{ \sigmax{\mu} \right\} _{\mu=1}^{N-1}\right\rangle$. After the measurements are performed the resulting quantum information is maintained on the output qubit which is the last qubit of the chain.

Logical operators of the resource state, $ P \in \mathcal{L} = \cent(\mathcal{R}) \backslash \mathcal{R}$, are measured if $P \sim P' \in \mathcal{M}$. For now we choose $\mathcal{M}^A = \left\langle \sigmax{a}\right\rangle$ to this end. In this example we couple the ancilla to the target qubit indexed $T = P(t)$ with unitary $U^A = U^Z[T,a] $ to perform a logical measurement $P \in \mathcal{L}$ where $P = X,\, Z$ is a logical Pauli operator.

We check that $P' \in \mathcal{M}$ by studying the stabilizer group of the resource state $\mathcal{R}$. The state has stabilizers
\begin{equation}
C[a] = \sigmax{a} \sigmaz{T},
\end{equation}
and
\begin{equation}
C[T] =  \sigmaz{T-1} \sigmax{T} \sigmaz{T+1} \sigmaz{a}.
\end{equation}

Now, using that $P \sim P' = \Sigma^P(t) \sigmaz{P(t)}$ we have 
\begin{equation}
P \sim P' C[a] = \Sigma^P(t) \sigmax{a} \in \mathcal{M},
\end{equation} 
provided $\mathcal{M}^A = \left\langle \sigmax{a} \right\rangle$ as prescribed, thus giving the desired logical measurement at time interval $t$ once the resource state is measured with measurement pattern $\mathcal{M}$.

\begin{figure}
\includegraphics{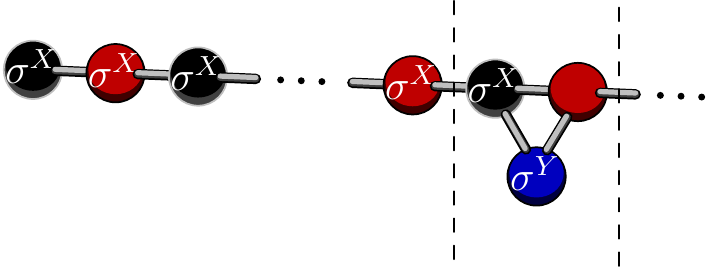}
\caption{\label{Fig:TypeIPauliY} The logical Pauli-Y operator on a type-I foliated qubit where the blue ancilla qubit is coupled both to qubits $Z(t)$ and $X(t)$ of the chain. The logical operator commutes with the measurement pattern where all of the qubits along the chain are measured in the Pauli-X basis, and the ancilla qubit is measured in the Pauli-Y basis, thus allowing us to recover Pauli-Y data from the foliated qubit.}
\end{figure}

We additionally find that we can measure the logical Pauli-Y information from a type-I foliated qubit using an ancilla-assisted measurement. We achieve this by coupling the ancilla to multiple target qubits. To show this we continue with the resource state model $\mathcal{R}$ given in Eqn.~(\ref{Eqn:Resource}) where again we have $\mathcal{A} = \left\langle \sigmax{a} \right\rangle$ and $\chan$ is a one-dimensional cluster state. The chain is measured as a type-I foliated qubit such that $\mathcal{M}^C =  \left\langle \left\{ \sigmax{\mu} \right\} _{\mu=1}^{N-1}\right\rangle$, and we couple the ancialla to the chain with the unitary $U^A = U^Z[Z(t),a] \times U^Z[X(t),a]$. We find that $Y \sim Y' \in \mathcal{M}$ provided $\mathcal{M}^A = \left\langle \sigmay{a}\right \rangle $ as we show below.

From the discussion given, we have that the operator $ \text{i} \left( \Sigma^X(t) \sigmaz{X(t)} \right) \left(\Sigma^Z(t) \sigmaz{Z(t)}\right)$ is a representative of the logical Pauli-Y operator for $\chan$. It is then readily checked then that
\begin{equation} 
 \text{i} X Z \sim Y' = \text{i} \left( \Sigma^X(t) \sigmaz{X(t)} \right) \left(\Sigma^Z(t) \sigmaz{Z(t)}\right) \sigmaz{a},
\end{equation}
is a representative of the logical Pauli-Y operator of $\mathcal{R}$ using the expressions given in Eqns.~(\ref{Eqn:CZPauliX}) and~(\ref{Eqn:SigmaTypeI}). Then using the stabilizer $C[a] \in \mathcal{R} $ where, 
\begin{equation}
C[a] = \sigmax{a} \sigmaz{Z(t)}\sigmaz{X(t)}, 
\end{equation}
we find that
$ \text{i} X Z  \sim Y' C[a] \in \mathcal{M}$
provided $\mathcal{M}^A = \left\langle \sigmay{a} \right\rangle$ which is depicted in Fig.~\ref{Fig:TypeIPauliY}, thus showing we can infer the logical Pauli-Y data of a type-I foliated qubit from $\mathcal{M}$ by coupling to multiple targets. We finally remark that one can show that the resource state with $U^A = U^Z[T_1,a] \times U^Z[T_2,a]$ where $T_1 = X(t_1)$ and $T_2 = Z(t_2)$ such that $t_1 \not= t_2$, we also recover encoded Pauli-Y information. We do not require this degree of generality here so we leave the proof of this fact as an exercise to the reader.

We conclude this subsection by summarising the differences between type-I foliated qubits considered here and the type-II qubits discussed in App.~\ref{App:TypeII}. Indeed, here we have shown that we can make a Pauli-Y measurement with a foliated qubit by coupling an ancilla to two target qubits. In contrast, following an argument similar to that given above, by measuring the physical qubits of a foliated chain in the Pauli-Y basis instead of the Pauli-X basis we find that we can measure the Pauli-Y operator by coupling an ancilla to a single qubit. This comes at the expense of including three qubits at each time interval instead of two, as is the case with type-I foliated qubits. As such, the physicist that is looking to perform a measurement-based experiment that demands a significant number of Pauli-Y measurements should decide carefully whether type-I or type-II foliated qubits are the most appropriate. Ultimately, the decision should depend on whether preparing physical qubits, or entangling interactions between pairs of physical qubits, is the most precious commodity of a given laboratory.

\subsection{Parity measurements with foliated qubits}
\label{Subsec:ParityMeasurement}

In general it will be necessary to make parity measurements between several foliated qubits that are encoded on different chains. We now specify the channel of a resource state consisting of several foliated qubits, together with an ancilla system that we use to make logical measurements. More precisely we consider the resource state $\mathcal{R}  = U^A \left(\chan \otimes \mathcal{A} \right), $ of $n$ foliated qubits. As in the previous section, for now we consider a single ancilla prepared in a known eigenstate of the Pauli-X basis. The channel of the resource state is such that
\begin{equation}
\chan = \bigotimes_{j=1}^n \chan_j, \label{Eqn:MultiQubitChannel}
\end{equation}
and $\chan_j = U_j \mathcal{I}_j$ is the stabilizer group of the $j$-th one-dimensional cluster state of length $N_j$ as defined in Subsec.~\ref{Subsec:OneDcluster}. It will also become helpful later on to define the unitary operator that entangles the initial state to give the channel system, namely
\begin{equation}
U^C = \bigotimes_{j=1}^n U_j, \label{Eqn:ChannelCoupling}
\end{equation}
where $U_j = \prod_{\mu=1}^{N-1} U_j^Z[\mu,\mu+1]$ and $U_j^Z[\mu, \mu+1]$ denotes the controlled-phase gate acting on qubits $\mu $ and $\mu+1$ on the $j$-th chain of qubits. This will be helpful when we consider variations on the initial state $\mathcal{I} = \bigotimes_{j=1}^n \mathcal{I}_j$. One could choose the length of each cluster arbitrarily but for simplicity we suppose that all the type-I chains have an equal length $N_j =2D+1$ where $D$ is the number of time intervals and we include an additional qubit at the end of the chain to support the output of each foliated qubit.

The logical operators of $\chan_j$ are denoted $X_j$ and $Z_j$, and its physical qubits are indexed $X_j(t)$ and $Z_j(t)$ according to Eqn.~(\ref{Eqn:Indices}) where indices have been appended. Again, we have logical operators $P_j = X_j,\, Z_j$ such that
\begin{equation}
P_j \sim \Sigma_j^P(t)\sigmaz{P_j(t)}, \label{Eqn:ChannelLogical}
\end{equation}
where the stabilizer equivalence relation is taken with respect to the stabilizer group $\chan$. For now we consider the ancilla system that includes only a single ancilla qubit, i.e., $\mathcal{A}= \left\langle \sigmax{a} \right\rangle $.

The measurement pattern for the foliated system $\mathcal{M} = \mathcal{M}^C\otimes \mathcal{M}^A$ is specified
\begin{equation}
\mathcal{M}^C = \bigotimes_{j = 1} ^n \mathcal{M}^C_j,
\end{equation} 
where $\mathcal{M}^C_j$ is the set of single-qubit measurements acting for the cluster $\chan_j$. We have that $\mathcal{M}^C_j = \left\langle \left\{ \sigmax{\mu} \right\}_{\mu = 1}^{N_j-1} \right\rangle$ for type-I foliated qubits. We determine the measurements we make on the ancilla system, $\mathcal{M}^A$, depending on the choice of parity measurement.

We look to prepare $\mathcal{R}$ such that we measure the logical Pauli operator $ P  \in \mathcal{P}_n$ by measuring the ancilla qubit in an appropriate basis. To achieve this logical measurement, we couple the ancilla to the channel with the unitary
\begin{equation}
U^A = \prod_{j \in \supp{p^X}} U^Z[X_j(t), a] \prod_{j \in \supp{p^Z}} U^Z[Z_j(t), a],
\end{equation}
where $p = (p^X \, p^Z)^T$ such that $p = v(P)$ as defined in Eqn.~(\ref{Eqn:Vectorise}). Again, for simplicity, we have coupled the ancilla to the physical qubits of a common time interval $t$, but showing our construction is general beyond this constraint is straight forward. In fact, as we will observe, we find practical benefits from coupling an ancilla to physical qubits in different time intervals later in Sec.~\ref{Sec:TwistedSurfaceCode}.

\begin{figure}
\includegraphics{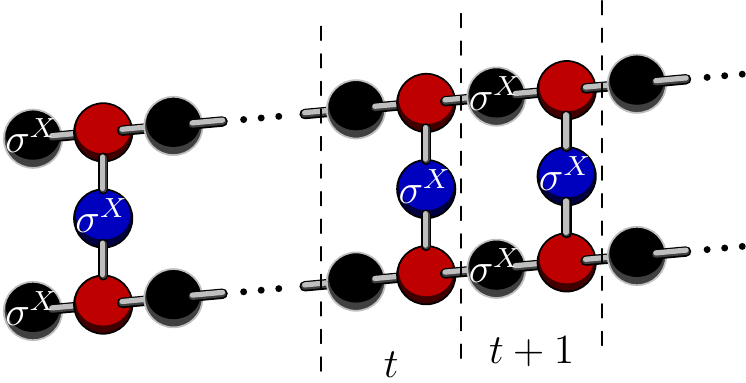}
\caption{\label{Fig:ParityMeasurements} Two type-I foliated qubits where ancillas are coupled to perform $X_1 X_2$ parity measurements. To the left of the figure we show a logical operator $\sim X_1 X_2$ by coupling an ancilla, shown in blue, to qubits $X_1(1)$ and $X_2(1)$. To the right of the figure we show an element of the measurement pattern that is generated by two parity measurements that are made at different time intervals. This term is also a stabilizer of $\mathcal{F}$.}
\end{figure}

Upon coupling the ancilla to the channel system to form the resource state, the new stabilizer group for the resource state will include $C[a] = U^A \sigmax{a} U^{A\dagger} \in \mathcal{R}$ such that
\begin{equation}
C[a] = \sigmax{a} \prod_{j \in \supp{p^X}} \sigmaz{X_j(t)} \prod_{j \in \supp{p^Z}}  \sigmaz{Z_j(t)}.
\end{equation}

We additionally have logical operators of the resource state
\begin{equation}
X_j \sim \sigmaz{a}^{p^Z_j} \Sigma_j^X(t)\sigmaz{X_j(t)} ,
\end{equation}
and
\begin{equation}
 Z_j \sim \Sigma_j^Z(t)\sigmaz{Z_j(t)}.
\end{equation}
which follows from Eqn.~(\ref{Eqn:CZPauliX}) and the definition of the logical operators of the channel system $\chan$ shown in Eqn.~(\ref{Eqn:ChannelLogical}) where $p^Z_j$ is the $j$-th element of the vector $p^Z$. Combining the above expressions we find $P' \sim P$ such that
\begin{eqnarray}
P' &=& \sigmaz{a}^{p^X\cdot p^Z}  \prod_{j \in \supp{p^X}} \Sigma_j^X(t)  \sigmaz{X_j(t)}  \\ && \times  \prod_{j \in \supp{p^Z}} \Sigma_j^Z(t)  \sigmaz{Z_j(t)}. \nonumber \label{Eqn:Centre}
\end{eqnarray}
It follows then that $P  \sim  P' C[a] $ such that
\begin{equation}
P' C[a] =  \sigmax{a} \sigmaz{a}^{p^X\cdot p^Z}  \prod_{j \in \supp{p^X}} \Sigma_j^X(t)  \prod_{j \in \supp{p^Z}} \Sigma_j^Z(t). \label{Eqn:CentreElement}
\end{equation}
We therefore have that $P \sim C[a] P' \in \mathcal{M}$ provided we choose $\sigmax{a} \sigmaz{a}^{p^X\cdot p^Z} \in \mathcal{M}^A $, thus showing we can infer the logical operator $P$ from the measurement data given an appropriate choice of measurements $\mathcal{M}$. To the left of Fig.~\ref{Fig:ParityMeasurements} we show the element $P' \in \mathcal{M}$ with $P' \sim P = X_1X_2$ measured from a pair of type-I foliated qubits. We also show a parity check $P' \in \mathcal{M}$  in Fig.~\ref{Fig:PauliYY} such that we measure $P' \sim P = Y_1 Y_2$. In this case the ancilla qubit couples with each foliated qubit twice to include Pauli-Y terms in the check.

\begin{figure}
\includegraphics{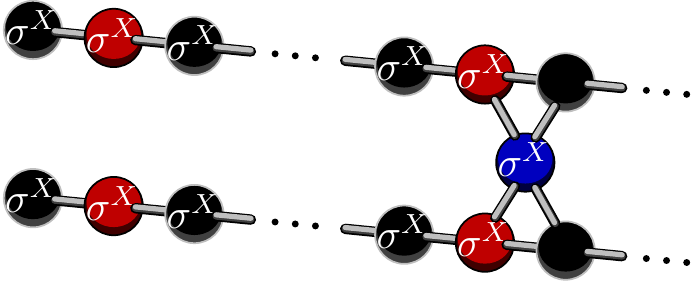}
\caption{\label{Fig:PauliYY} Two type-I foliated qubits where we use the ancilla to measure the parity of the two foliated qubits in the Pauli-Y basis. The ancilla is coupled to the physical qubits indexed $X_j(t)$ and $Z_j(t)$ for each foliated qubit $j = 1,\,2$. We measure the ancilla in the Pauli-X basis, since there are an even number of Pauli-Y terms in the parity measurement. }
\end{figure}

\subsection{The compatibility of parity measurements}
\label{Subsec:Defrustation}

Having now discussed how to include a single logical measurement in the foliated system, we next investigate the conditions under which we can simultaneously measure two degrees of freedom of the input system, $P$ and $Q$ where $P,\,Q \in \mathcal{P}_n$.

We now consider a resource state $\mathcal{R}$ with two ancillas $\mathcal{A} = \left\langle \sigmax{a}, \sigmax{b} \right\rangle $ that are coupled to the channel $\chan = \bigotimes_{j} \chan_j$ at time intervals $t$ and $t'$ via unitary $U^A = U_a U_b$ where
\begin{equation}
U_a =  \prod_{j \in \supp{p^X}} U^Z[X_j(t), a] \prod_{j \in \supp{p^Z}} U^Z[Z_j(t), a],
\end{equation}
and
\begin{equation}
U_b =  \prod_{j \in \supp{q^X}} U^Z[X_j(t'), b] \prod_{j \in \supp{q^Z}} U^Z[Z_j(t'), b].
\end{equation}
where $p = v(P)$ and $q = v(Q)$. The measurement pattern is chosen according to the rules above such that the chosen parity measurements are made correctly. In any case ancillas $a$ and $b$ are measured in either the Pauli-X or Pauli-Y basis. We first suppose that $t' \not= t$, and afterwards we look at the more complicated case where $t = t'$. As we will show, we find that we can can measure both $P$ and $Q$ simultaneously from the logical space of the channel provided $P$ and $Q$ commute. While the setup presented above is sufficient if both $U_a$ and $U_b$ are coupled to the channel at two different time intervals, we find that in certain situations it is necessary to modify $U^A$ in order to measure both $P$ and $Q$ with a coupling at a common time interval.

Without loss of generality we begin with the case where $t < t'$ such that $P'$ as in Eqn.~(\ref{Eqn:CentreElement}) is an element of $\mathcal{R}$ since this operator shares no common support with $U_b$. The resource state additionally includes the representative $Q'\sim Q$ of the logical operator, such that
\begin{equation}
Q' = U_a \left( M[b] \prod_{j \in \supp{q^X}} \Sigma_j^X(t')  \prod_{j \in \supp{q^Z}} \Sigma_j^Z(t') \right) U_a^\dagger
\end{equation}
where we take $M[b] = \sigmax{b} \sigmaz{b}^{q^X\cdot q^Z} $ which follows from the inclusion of the term $C[b] = U_b \sigmax{b}U_b^\dagger \in \mathcal{R}$. To determine the value of $Q'$ we are interested in the commutation relations between the operator $U_ a$ and the operators $\Sigma_j^X(t')$ and $\Sigma_j^Z(t')$ which share mutual support on the physical system. Using Eqns.~(\ref{Eqn:CZPauliX}) and~(\ref{Eqn:ChannelLogical}) it follows that
\begin{equation}
U_a \Sigma^X_j(t') U_a^\dagger = \Sigma^X_j(t')  \sigmaz{a}^{p^Z_j},
\end{equation}
and
\begin{equation}
U_a \Sigma^Z_j(t') U_a^\dagger = \Sigma^Z_j(t')  \sigmaz{a}^{p^X_j}.
\end{equation}
Given that $U_a M[b] U_a^\dagger = M[b]$ we thus have that 
\begin{equation}
Q' =  \sigmaz{a}^{\Upsilon(p,q)}  M[b] \prod_{j \in \supp{q^X}} \Sigma_j^X(t')  \prod_{j \in \supp{q^Z}} \Sigma_j^Z(t'),
\end{equation}
where $\Upsilon(p,q)$ is given in Eqn.~(\ref{Eqn:Upsilon}). In this case, for $ \Upsilon(p,q) \not= 0 $ we must measure ancilla $a$ in the Pauli-Z basis in order to infer the value of $Q$ from the measurement pattern. On the other hand, to measure $P$ we must measure ancilla $a$ in either the Pauli-X basis or the Pauli-Y basis to infer its value from $\mathcal{M}$. The conclusion of this discussion is that we cannot infer both measurements $P$ and $Q$ from $\mathcal{M}$ unless $\Upsilon(p,q) = 0$ in order for them to be measured simultaneously from the logical space of the channel. This is, of course, consistent with the standard postulates of quantum mechanics which only permits the simultaneous measurement of both $P$ and $Q$ provided the operators commute.

We next consider the case that $t = t'$. In this case we have that
\begin{equation}
U_a \Sigma^X_j(t) U_a^\dagger = \Sigma^X_j(t) \sigmaz{a} ^ {p^Z_j}, 
\end{equation}
and
\begin{equation}
U_b \Sigma^X_j(t) U_b^\dagger = \Sigma^X_j(t) \sigmaz{b} ^ {q^Z_j}.
\end{equation}
Unlike the previous case though, $U_a$ and $U_b$ commute with $\Sigma^Z_j(t)$ terms. To this end we find that
\begin{equation}
P \sim  M[a] \sigmaz{b}^{p^X \cdot q^Z}  \prod_{j \in \supp{p^X}} \Sigma_j^X(t)  \prod_{j \in \supp{p^Z}} \Sigma_j^Z(t),
\end{equation}
and
\begin{equation}
Q \sim  M[b] \sigmaz{a}^{p^Z \cdot q^X}  \prod_{j \in \supp{q^X}} \Sigma_j^X(t')  \prod_{j \in \supp{q^Z}} \Sigma_j^Z(t'),
\end{equation}
where $M[a]$ and $M[b]$ are either Pauli-X or Pauli-Y measurements. We thus see that with the current choice of $U^A$ the measurements of both $P$ and $Q$ are incompatible unless both $p^X \cdot q^Z = 0$ and $p^Z \cdot q^X= 0 $; a sufficient condition for $P$ and $Q$ to commute. Indeed, if $PQ = QP $ then $p^X \cdot q^Z = p^Z \cdot q^X$. We note that for CSS stabilizer codes where either $p^X = 0$ or $p^Z=0$ holds for all stabilizers, we have that $p^X\cdot q^Z = p^Z \cdot q^X = 0$ necessarily holds between all pairs of stabilizers. As such, none of the following measures we describe are necessary to measure the stabilizer checks for foliated CSS codes as the authors considered in Ref.~\cite{Bolt18}.

From here on we suppose that $P$ and $Q$ commute. Supposing that $p^X \cdot q^Z = p^Z \cdot q^X = 0$ we find that the values of $P$ and $Q$ are members of $\mathcal{M}$ as we have already proposed. We find that we can also modify the channel to measure two commuting operators $P$ and $Q$, even if $p^X \cdot q^Z = p^Z \cdot q^X = 1$. One approach to deal with this issue, as we have already discovered, is to measure the two operators at different time intervals. A smaller modification may simply be to change only a subset of the target qubits of either $U_a$ or $U_b$ onto a different time interval. We present an example where we use this method to good effect in the following Section. Alternatively, we can append an additional controlled-phase gate, $U^Z[a,b]$, in the entangling circuit $U^A$ which also recovers the compatibility of the two commuting measurements. We adopt this approach throughout the main text.

\section{The foliated system}
\label{Sec:CodeFoliation}

In this Section we define the foliated channel; the fault-tolerant system that will propagate an input quantum error-correcting code onto an output state. Further, we will also see that with a suitable choice of resource state we are able to measure the input code in such a way that it is deformed by the foliated system. In what follows we sketch out the different components we use to make up a foliated system before presenting two theorems that describe its function. 
Specifically, we will require definitions of the initial system, the channel system, the resource state and the system after the the prescribed pattern of measurements is made. Theorem~\ref{Thm:OutputCode} then explains the action of the foliated system on the logical input state, and Theorem~\ref{Thm:ChannelStabilizers} describes stabilizer checks we can use to identify errors that act on the physical qubits of the foliated system. See Ref.~\cite{Vuillot18} for a related discussion from the perspective of code deformations with stabilizer codes. The construction is summarised in Fig.~\ref{Fig:Sketch}.

\begin{figure}
\includegraphics{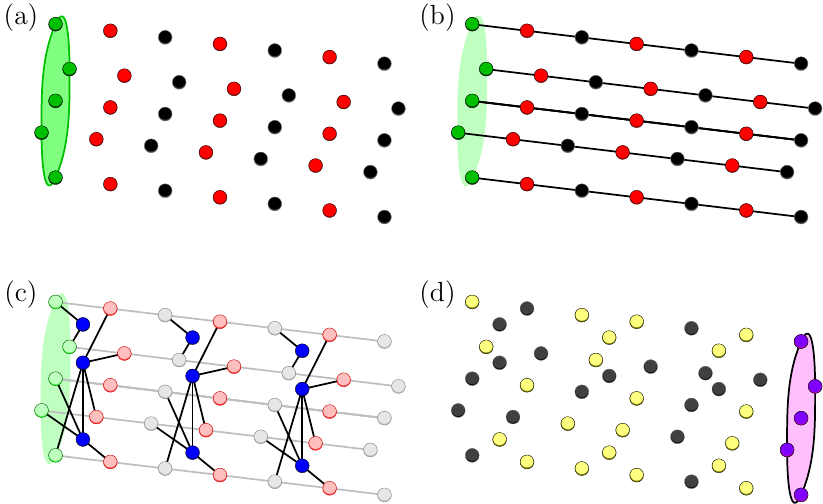}
\caption{Sketch of the foliation process and construction of the system. Time moves from left to right across the page in each figure. (a)~The initial system, $\mathcal{I}$, consists of an input code, $\mathcal{G}_{\text{in}}$, shown with green qubits to the left of the figure, and other qubits, shown in black and red in the figure, prepared in the $|+\rangle$ state. (b)~The channel system $\mathcal{K}$ is produced by applying unitary $U^C$ to $\mathcal{I}$. The edges in the figure represent controlled-phase gates prepared between pairs of qubits. Effectively, we have concatenated each of the qubits of $\mathcal{G}_{\text{in}}$ into the code space of a one-dimensional cluster state. (c)~We produce $\mathcal{R}$ by entangling the qubits of the ancilla system, $\mathcal{A}$, to the channel system. The coupling is specified by the choice of $\mathcal{G}_{\text{ch.}}$. This choice determines the check operators of the foliated system, and the deformation on the input qubit. (d)~The qubits of $\mathcal{R}$ are measured according to $\mathcal{M}$. This propagates the input system onto the output system, $\mathcal{G}_\text{out}$, shown in purple to the right of the figure. The measurement pattern determines the value of the stabilizer checks, and in turn the correction that must be applied to the output system. \label{Fig:Sketch}}
\end{figure}

As described in Sec.~\ref{Sec:FoliationModel}, the foliated system is defined as 
\begin{equation}
\mathcal{F} = \mathcal{R} \cup \mathcal{M},
\end{equation} 
where $\mathcal{R}$ is the resource state, and $\mathcal{M}$ is the measurement pattern. The system $\mathcal{F}$ is jointly determined by an input code $\mathcal{G}_{\text{in}}$ and a channel code $\mathcal{G}_\text{ch.}$. Also to be determined is the length of the channel $D$, which, unless otherwise stated, we suppose is large, i.e. comparable to the distance of the code $\mathcal{G}_{\text{ch.}}$. We assume that both codes $\mathcal{G}_{\text{in}}$ and $\mathcal{G}_{\text{ch.}}$ are stabilizer codes and, moreover, are supported on the same set of qubits that we index $1 \le j \le n$.

The resource state is of the form $\mathcal{R} = U^A \left(  \chan \otimes \mathcal{A} \right)$ where $\chan$ is the channel system and $\mathcal{A}$ is a set of ancilla qubits prepared in a product state that we couple to the channel with unitary operator $U^A$. The channel $\chan$ is specified by $\Gin$. The ancilla system $\mathcal{A}$ and the entangling unitary $U^A$ are determined by a generating set of $\Gch$.

We remind the reader that we index the qubits of each chain in terms of time intervals, indexed $1\le t \le D$, such that the qubit at the $2 t -1$-th site ($2 t$-th site) of the $j$-th chain is indexed $Z_j(t)$ ($X_j(t)$) and we have that the last qubit of each chain is indexed $Z_j(D+1) = 2D+1$ in the Pauli coordinate system which supports the output state.

Previously we have considered the tensor product of $n$ one-dimensional cluster states, $ \bigotimes_{j=1}^n \mathcal{K}_j$, where $\chan_j = U_j \mathcal{I}_j$, to propagate $n$ individual qubits described by logical operators $X_j$ and $Z_j$ for $1 \le j \le n$. In the following definition of the channel system we encode the input code $\Gin$ onto the logical qubits of the $n$ foliated qubits.

\begin{defn}[Channel system]
\label{Def:Channel}
The channel system $ \chan = U^C \mathcal{I} $ is produced by applying unitary operator $U^C $ as defined in Eqn.~(\ref{Eqn:ChannelCoupling}) to the initial state $\mathcal{I}$. The stabilizer group $\mathcal{I}$ is such that the stabilizers of $\Gin$ are encoded on the qubits indexed $Z_j(1)$, and the other qubits are prepared in a product state. Explicitly, for each $G \in \Gin$ we have $G_{\text{in}} \in \mathcal{I}$ with
\begin{equation}
G_{\text{in}} \equiv \prod_{j\in \supp{g^X}} \sigmax{Z_j(1)}  \prod_{j\in \supp{g^Z}} \sigmaz{Z_j(1)},
\end{equation}
where $(g^X \, g^Z)^T = v(G)$. The other qubits of $\mathcal{I}$ are encoded in the $+1$ eigenvalue eigenstate of the Pauli-X matrix, such that we have
\begin{equation}
\sigmax{X_j(t)} \in \mathcal{I} \quad \forall j,\, t \le D, 
\end{equation}
and
\begin{equation}
\sigmax{Z_j(t)} \in \mathcal{I} \quad \forall j, \, 2 \le  t \le D +1. 
\end{equation}
\end{defn}

The resource state is generated by coupling ancilla qubits to the channel. Measuring these qubits in the appropriate basis provides data to identify qubits that have experienced errors. We couple ancilla to the resource state in a fashion according to a particular generating set $\Gres$ of some Abelian stabilizer group $\Gch$. The generating set $\Gres$ may be an over complete generating set of $\Gch$.

\begin{defn}[Resource state] \label{Def:Resource}
For a choice of input code $\Gin$ that is implicit in the channel system $\chan$, see Def.~\ref{Def:Channel}, and channel code $\Gch$ with a specified generating set $\mathcal{G}_{\mathcal{R}}$, we define the stabilizer group of the resource state by
	\begin{equation}
	\mathcal{R} = U^A  \left(  \chan \otimes \mathcal{A} \right), 
	\end{equation}
	where the ancilla system $\mathcal{A}$ and entangling unitary $U^A$ are defined as follows. The ancilla system is in the product state
	 \begin{equation}
	\mathcal{A}  = \left\{  \sigmax{G(t)} : \forall t,\, G \in \mathcal{G}_\mathcal{R} \right\},
	 \end{equation}
where the coordinates $G(t)$ uniquely index all the ancillae in the ancilla system at a given time interval $t$. The entangling unitary $U^A$ is given by
	\begin{equation}
	U^A = V  \prod_{G\in\Gres, \, t}  U[G(t)] , \label{Eqn:UA}
	\end{equation} 
	where 
	\begin{equation}
	U[G(t)] = \prod_{j \in \supp{g^X}} U^Z[X_j(t), G(t)] \prod_{j \in \supp{g^Z}} U^Z[Z_j(t), G(t)], 
	\end{equation}
	for $(g^X \, g^Z)^T = v(G)$ and $V = \prod_t V(t)$ with
	\begin{equation}
	V(t) = \prod_{\langle\langle G, H \rangle \rangle} U^Z\left[G(t) , H(t)\right]^{g^X\cdot h^Z},
	\end{equation}
where we take the product of all unordered pairs of elements $G,\,H \in \Gres$ with $G \not=H$, and we have that $(g^X\,g^Z)^T = v(G)$ and $(h^X\,h^Z)^T = v(H)$.
\end{defn}

We remark that in the above definition, the unitary $U^A$ couples the ancilla qubits to the channel system along with each other. The component $V$ enables the simultaneous measurement of all commuting measurements of $G\in \mathcal{G}_{\mathcal{R}}$ at common time intervals. This is achieved by applying controlled-phase gates between ancilla qubits $G(t)$ and $H(t)$ where $g^X \cdot h^Z = g^Z \cdot h^X = 1$. A simple example giving intuition for why this operator is necessary is discussed in SubSec.~\ref{Subsec:Defrustation}.

We specify a measurement pattern that propagates information through the resource state onto the output system while additionally acquiring data that enables us to identify the locations of errors.

\begin{defn}[Measurement pattern] \label{Def:Measure}
The measurement pattern is such that
\begin{equation}
\mathcal{M}  =\mathcal{M}^C \cup \mathcal{M}^A,
\end{equation} 
where $\mathcal{M}^C$ and $\mathcal{M}^A$ denote the measurements on the channel and ancilla system, respectively, given by
\begin{equation}
	\mathcal{M}^C = \left\{ \sigma^X[Z_j(t)],\,\sigma^X[X_j(t)] \in \mathcal{M}^C: 1\leq t \leq D \right\},
	\end{equation}
	and
	\begin{equation}
	\mathcal{M}^A = \left\{M[G(t)] \in \mathcal{M}^A : G\in \mathcal{G_{\mathcal{R}}}, 1\leq t \leq D\right\},
\end{equation}
	where $(g^X \, g^Z)^T = v(G)$ and we have
\begin{equation}
M[G(t)] = \sigmax{G(t)} \left(\text{i}\sigmaz{G(t)}\right)^{g^X \cdot g^Z}. \label{Eqn:Measure}
\end{equation}
\end{defn}

It may be helpful to expand Eqn.~(\ref{Eqn:Measure}) such that
	\begin{equation}
	M[G(t)] = \begin{cases}
	\sigma^X[G(t)] \quad \text{if} ~ g^X\cdot g^Z = 0, \\
	\sigma^Y[G(t)] \quad \text{if}  ~g^X\cdot g^Z = 1,
	\end{cases}
	\end{equation}
and we remember the inner product is taken modulo 2.

\subsection{Foliation}

With the resource state $\mathcal{R}$ and measurement pattern $\mathcal{M}$ defined above we now turn our attention to the foliated system $\mathcal{F} = \mathcal{R} \cup \mathcal{M}$. Important properties of the system will be determined from elements in $\cent(\mathcal{F})$. While we have introduced the foliated model as a subsystem code, we have included additional structure to capture the process of foliation, namely, we prepare the system in a fixed gauge of the resource state, $\mathcal{R}$ and we project the system onto a gauge of $\mathcal{M}$. As such, elements of $\cent(\mathcal{F})$ have a different role depending on their inclusion in $\mathcal{R}$ and $\mathcal{M}$.

Elements $\cent(\mathcal{F}) \cap \mathcal{M}$ are observable degrees of freedom that are measured under projection by the single-qubit measurement pattern. Among these include stabilizer operators and logical operators of the input state that are measured by the resource. On the other hand elements $\cent(\mathcal{F}) \backslash \mathcal{M} $ give rise to the stabilizers and logical operators of the output state which are propagated or inferred at a later point by taking the output as the input of another channel. Further, we can also look at elements of $\cent(\mathcal{F})$ with respect to their membership of $\mathcal{R}$. Indeed, elements $\cent(\mathcal{F})\cap \mathcal{R} $ are stabilizers whereas elements $\cent(\mathcal{F})  \backslash \mathcal{R}$ are logical degrees of freedom that are either propagated through the resource state or measured under the projection. 

Remarkably, due to the decomposition of the foliated system we have presented here we can separate error correction, determined by stabilizer group $\mathcal{S}$, and the logical function of a given channel, $\Gout$, into two separate parts. Let us now characterize the output state of the foliated system after measurements have been performed, as well as the stabilizers of the channel responsible for detecting errors. 

The output of the channel consists of an encoding, determined by $\Gout$, and a set of logical operators $\mathcal{L}_{\text{out}}$. The logical function of the channel can be summarized by how it maps input logical operators to output logical operators. The following theorem describes the output encoding and logical degrees of freedom. 

\begin{thm}\label{Thm:OutputCode}
	For any foliated channel $\mathcal{F}$ determined by input code $\Gin$ and channel code $\Gch$, the output state is a codeword of the output code $\Gout$, with
	\begin{equation}
	\Gout = \Gch \cup \left( \Gin \cap \cent\left( \Gch \right) \right) . \label{Eqn:OutputStabilizer}
	\end{equation}
	The logical operators of $\Gout$ are given by
	\begin{equation}
	\mathcal{L}_{\text{out}} = \left( \cent(\Gin) \cap \cent(\Gch) \right) \backslash \left( \Gin \cup \Gch \right). \label{Eqn:OutputLogical}
	\end{equation}
	Further, elements of $\left(\cent(\Gin) \backslash \Gin\right) \cap \Gch$ are measured.
\end{thm}

The two equations given above specify precisely the function of a channel at a macroscopic level, independent of the foliated system that performs the manipulation of the input state.

For the channel to be fault-tolerant, we need the channel to contain stabilizers that can check for errors, and these stabilizers need to be able to be inferred from measurements. Recall the stabilizer of the foliated sytem is given by 
\begin{equation}
\mathcal{S} = \cent(\mathcal{F}) \cap \mathcal{R} \cap \mathcal{M}.
\end{equation} 
The following theorem identifies two types of important operators that we call bulk stabilizers and boundary stabilizers.  

\begin{thm}\label{Thm:ChannelStabilizers}
For any foliated channel $\mathcal{F}$ specified by input code $\Gin$ and $\Gres$ which generates $\Gch$, we have
 \begin{enumerate}
 	\item $S_{\text{bulk}}[G(t)] \in \mathcal{S}$ \quad $\forall G  \in \Gch,\,  2 \le  t \le D, $
 	\item $S_{\text{bdry.}} [G(t)] \in \mathcal{S}$ \quad  $\forall G \in \Gch \cap \Gin, \, t \le D,$
 \end{enumerate}  
where for $v(G) = (g^X \, g^Z)^T$ we have
\begin{eqnarray}
S_{\text{bulk}}[G(t)] & =&  M[G(t)] M[G(t-1)] \label{Eqn:StabilizerCell} \\ 
&&\times \prod_{j \in \supp{g^X}} \Sigma_j^X(t) \Sigma_j^X(t-1) \nonumber \\
&&\times \prod_{j \in \supp{g^Z}} \Sigma_j^Z(t) \Sigma_j^Z(t-1), \nonumber
\end{eqnarray} 
and
\begin{equation}
S_\text{bdry.}[G(t)]  =  \prod_{\tilde{G}\in \xi(G)}M[\tilde{G}(t)]  \prod_{j \in \supp{g^X}} \Sigma_j^X(t)  \prod_{j \in \supp{g^Z}} \Sigma_j^Z(t), \label{Eqn:BdryStabilizer}
\end{equation} 
where $\xi(G) \subseteq \Gres$ is defined such that the product of all the terms of $\xi(G)$ give $G \in \Gch$.
\end{thm}
We remark that the set $\xi(G) $ necessarily exists by definition, since we construct the resource state with terms $G \in \Gres$ which generate $\Gch$. The set $\xi(G)$ is not necessarily unique, in which case any choice will suffice. We defer the proofs of the above Theorems to Appendix~\ref{App:ThmProofs}.

It is also worth pointing out that Theorem~\ref{Thm:ChannelStabilizers} may not necessarily describe the stabilizer group of $\mathcal{F}$ exhaustively. Indeed, with certain choices of $\Gres$ we can obtain additional elements of $\cent(\mathcal{F}) \cap \mathcal{F}$. For instance, we may choose to foliate a self-correcting stabilizer model~\cite{Brown16} such as the four-dimensional toric code~\cite{Dennis02}. In which case we have additional checks that are local to a given time interval due to constraints among the stabilizer group. It is precisely the constraints that these models present that give rise to single-shot error correction~\cite{Bombin15, Brown16a, Duivenvoorden17, Campbell18, fawzi2018constant} which enable us to identify errors on the ancillary qubits we use to make checks, see also Appendix~\ref{App:SubsystemFoliation} for a discussion on single-shot error correction by foliation of the gauge color code. The stabilizers we have described in Theorem~\ref{Thm:ChannelStabilizers} are generic to all foliated models using the construction we have presented.

Given some foliated system, it is then important to find an efficient decoding algorithm. In recent work~\cite{Brown19} it was proposed that a decoder based on minimum-weight perfect matching can be found if the symmetries of a stabilizer code are known. This has been demonstrated for several examples of topological codes, see e.g.~\cite{Wang03}, and it is suggested that this idea can be extended to more general classes of stabilizer codes. Given such a decoder, Ref.~\cite{Brown19} also outlines how to extend these symmetries along a time-like direction to find a minimum-weight perfect-matching decoder for the situation where stabilizer measurements are unreliable, see Appendix D Example~8. This decoder will map directly onto the stabilizers of a foliated system where errors on ancilla qubits correspond to measurement errors. As such, we expect that this approach to designing decoding algorithms is readily applicable to the foliated models we have proposed here. We also remark on Ref.~\cite{Bolt18} where they propose a decoder for foliated CSS codes where a decoder for $\Gch$ is already known and belief propagation between time intervals is used.

Given a decoder, we can then compare the tolerance of different foliated systems to noise. Introducing an independent and identically distributed noise model to the foliated systems we have proposed will be equivalent to a phenomenological noise model~\cite{Wang03} for the channel code. With some thought, one could also consider introducing correlated errors between the qubits of the foliated systems to simulate circuit level errors on some system of interest. This was first studied with the surface code in Ref.~\cite{Raussendorf06}.

In the following section we apply our procedures for foliation to the twisted surface code model as an illustrative example that we later use in our model of fault-tolerant quantum computation. We note that the above construction and theorems are readily generalized to include other methods of foliation including foliation using type-II qubits as we will see in the next section. Our scheme for foliation and related results also generalize to subsystem codes provided they are generated by a group with particular properties. For instance, our results apply directly to all CSS subsystem codes. We give a discussion of this extension in Appendix~\ref{App:SubsystemFoliation}.

We finally remark that our model for foliation shares some parallels with the work in Ref.~\cite{Bacon15}. In this work, the authors find subsystem codes by unifying the wires of a circuit of Clifford gates with the qubits of a quantum error-correcting code. They then determine gauge generators from the Clifford operations of the circuit. In our work the cluster-state model naturally separates the wire segments of a quantum computation and specifies them by qubits over time intervals. We then include an additional qubit for each measurement we make over the duration of the computation. The difference between these schemes is that the former construction focuses on the unitary gates of an encoding circuit whereas here we focus on the measurements we use to detect errors. While the two constructions are apparently distinct, it may be valuable to find a way of unifying them in future work.

\section{The twisted surface code}
\label{Sec:TwistedSurfaceCode}

The twisted surface code~\cite{Yoder17} provides an interesting example of a stabilizer code that cannot be foliated using existing constructions~\cite{Bolt16}. Moreover this example will be helpful later on when we consider different schemes for fault-tolerant measurement-based quantum computation. We will briefly review its stabilizer group before showing resource states for the foliated model.

The stabilizer group for the twisted surface code is closely related to that of the surface code. However the model has an improved encoding rate due to a twist defect~\cite{Bombin10} lying in the centre of the lattice~\cite{Brown17, Yoder17}. Nonetheless, due to the central twist defect we cannot find a CSS representation of the model, and as such, we use the generalized construction for foliation given above. We will mostly focus on these stabilizers.

\begin{figure}
\includegraphics{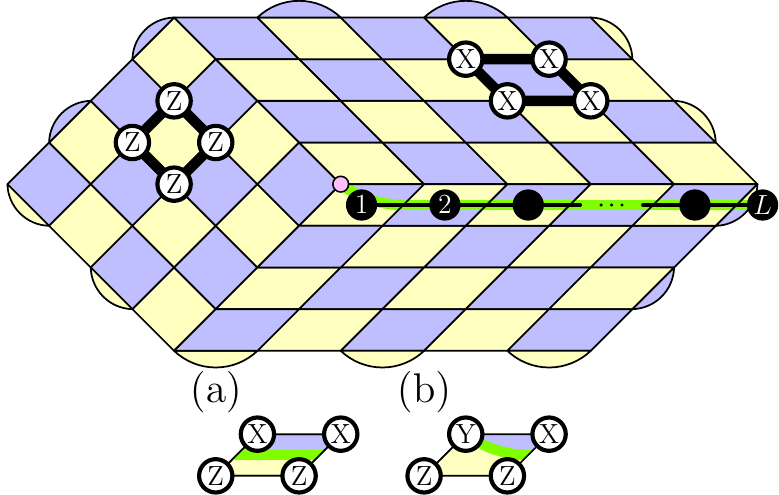}
\caption{\label{Fig:Twisted} The twisted surface code model. Examples of star and plaquette operators are shown explicitly on thick outlined plaquttes on blue and yellow faces respectively. Logical operators are supported on the qubits covered by the yellow and blue dotted lines. A defect line runs from the right-hand side of the lattice to its centre. An example of a modified stabilizer that supported on the defect line is shown at~(a). The green line terminates at the centre of the lattice. We write the stabilizer where the defect line terminates below the lattice at~(b). We focus on foliating the stabilizers that lie on the defect line, as such we number them explicitly with indices shown in black circles on the figure.}
\end{figure}

The stabilizer group for twisted surface code $\mathcal{G}_{\text{twisted}}$ is represented on the lattice in Fig.~\ref{Fig:Twisted} where qubits lie on the vertices of the lattice and stabilizers are associated to the faces, indexed $f$. On blue(yellow) faces we have the well-known star(plaquette) operators which are the product of Pauli-X(Pauli-Z) operators lying on the vertices on the corners of their respective face. The twisted surface code additionally has a defect line running from the right-hand side of the lattice to the centre, along which, stabilizers are modified. The stabilizers lying along the defect line, which is marked green in the figure, are weight-four terms that are the product of Pauli-X(Pauli-Z) stabilizers on the blue(yellow) part of the face. An example of a modified stabilizer is shown in Fig.~\ref{Fig:Twisted}(a). The stabilizer where the defect line terminates also includes a Pauli-Y term as shown in Fig.~\ref{Fig:Twisted}(b). Representations of logical operators $\overline{X}$($\overline{Z}$) that act on the one encoded qubit are the product of Pauli-X(Pauli-Z) operators supported on the dotted blue (yellow) lines on the lattice.

We now consider foliating the twisted surface code with the prescription given in the previous section where $\Gch = \mathcal{G}_{\text{twisted}}$. Stabilizers $G \in \mathcal{G}_{\text{twisted}}$ of the form of a CSS code, i.e. stabilizers where either $g^X$ or $g^Z $ are the null vector where $(g^X \, g^Z)^T = v(G)$, are foliated using the methods of Ref.~\cite{Bolt16} and take the form of those in the original works~\cite{Raussendorf06,Raussendorf07,Raussendorf07a}, as such we focus on the terms that lie along the defect line. In Fig.~\ref{Fig:CoupledAncilla} we show the resulting graph where physical qubits indexed $Z_j(t)$($X_j(t)$) are colored black(red) in the figure and time runs vertically up the page.

We denote the stabilizers lying along the defect line $G_f \in \mathcal{G}_{\text{twisted}}$ where we index face terms $1 \le f \le L$ as shown in Fig.~\ref{Fig:Twisted}. We consider vectors $ v(G_f) = (g_f^X \, g_f^Z)^T $. Notably, the stabilizers that lie along the defect line all have $g^X_f \cdot g^Z_{f+1} = 1 $. As such the operator
\begin{equation}
V = \prod_{f, t} U^Z[S_f(t), S_{f+1}(t)], 
\end{equation}
couples adjacent ancilla qubits. These bonds are shown as horizontal edges connecting adjacent blue ancilla qubits in Fig.~\ref{Fig:CoupledAncilla}. Further, the stabilizer $G_1$ is such that $g_1^X \cdot g_1^Z = 1$, we therefore measure these ancillas in the Pauli-Y basis, i.e., $\sigmay{G_1(t)} \in \mathcal{M}$. We color these ancilla qubits in green in Fig.~\ref{Fig:CoupledAncilla}.

The figure also shows two stabilizers, i.e., elements of $\mathcal{S} = \cent(\mathcal{F}) \cap \mathcal{R} \cap \mathcal{M}$ of the form shown in Eqn.~(\ref{Eqn:StabilizerCell}). The stabilizers are weight six, except the stabilizer at the centre of the lattice which includes a Pauli-Y term. This operator is weight seven.

We finally remark that the resource state we have thus far considered to carry out its designated function is by no means unique, and using other foliation techniques considered in the previous section we can devise different, arguably favourable, resource states. In Fig~\ref{Fig:LiftedChecks} we show one such alternative. In this figure, for even values of $f$ we replace the coupling operators as follows
\begin{eqnarray}
U[G_f(t)] &\rightarrow& \prod_{j \in \supp{g^X_f}} U^Z[X_j(t), G_f(t)] \nonumber \\ && \times \prod_{j \in \supp{s^Z_f}} U^Z[Z_j(t+1), G_f(t)], 
\end{eqnarray}
where now the check couples to targets $Z_j(t+1)$ instead of $Z_j(t)$. With this modification all the check measurements are compatible where we set $V=1$, thus reducing the valency of the ancilla qubits. We additionally replace the foliated qubit where the defect terminates with a type-II foliated qubit, which is colored by blue edges. This further reduces the valency of ancillas used to measure operators $S_1(t)$ as we need only couple to a single target to measure the Pauli-Y component of this check. This also reduces the valency of the physical qubits of the type-II foliated chain. However, this comes at the expense of including an additional physical qubit per time interval to include a type-II foliated qubit in the system.

To summarize the present discussion we have seen that we can introduce Pauli-Y terms in stabilizer measurements either by using type-II foliated qubits or by coupling an ancilla to multiple targets of the same foliated qubit, and we have seen that in general we can periodically measure a full set of check measurements by either lifting the targets of selected check measurements to higher time intervals or by coupling pairs of ancillas where it is appropriate. The examples of resource states we have considered here are by no means exhaustive and other variations can be made following the general principles of foliation given in the previous section, but given the multitude of variations one could come up with we leave further experimentation to the reader. One variation that is worth mentioning~\cite{Bombin17} is where additional checks are made at intermediate time intervals where we couple extra ancilla to physical qubits at layers indexed $X_j(t)$ and $Z_j(t+1)$. Of course, $V$ must also be modified to account for this change. The new foliated system requires more ancilla qubits along the defect line, and the valency of the physical qubits along each foliated wire is also increased, but in return we reduce the weight of the check operators and increase the number of available check operators. We discuss this compressed construction in Appendix~\ref{App:Compressed}.

\begin{figure}
	\includegraphics{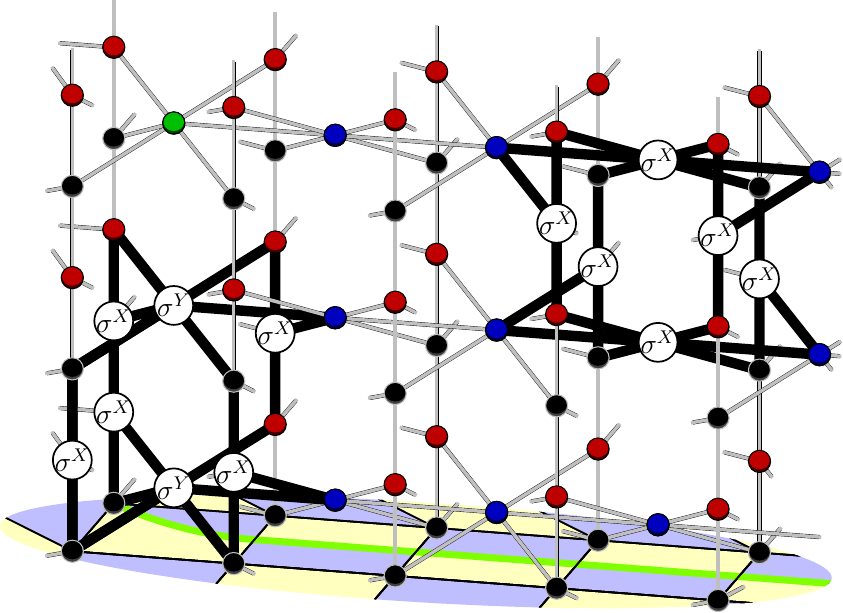}
	\caption{\label{Fig:CoupledAncilla} A foliation pattern for the defect line of the twisted surface code. The code is initialized at the bottom of the figure where the lattice that describes the stabilizers of the code are shown. Qubits that are coupled with a controlled-phase gate are connected by an edge. Qubits indexed $Z_j(t)$($X_j(t)$) are colored black(red), and ancilla qubits are colored in blue(green), if they are measured in the Pauli-X(Pauli-Y) basis. Example elements $S_{\text{bulk}}$ (as in Theorem~\ref{Thm:ChannelStabilizers}) which belong to $ \mathcal{S} =  \cent (\mathcal{F}) \cap \mathcal{R} \cap \mathcal{M}$ are also shown.}
\end{figure}

\section{Quantum computation with the foliated surface code}
\label{Sec:FTQC}

In what follows we explore the general theory we have established here by following the example of the foliated surface code. This model is a CSS stabilizer code and is thus foliated using the methods in Ref.~\cite{Bolt16} and has been studied from the perspective of computation in Refs.~\cite{Raussendorf07, Raussendorf07a,  Herr17b}. Beyond the work in the literature, using the generalized framework for foliation we have presented, we can additionally realize a phase gate deterministically in the foliated picture using ideas from Refs.~\cite{Brown17, Yoder17}. In particular, in Ref.~\cite{Brown17} it was shown that the corners of the planar code where two distinct types of boundary meet can be regarded as a Majorana mode. Throughout the discussion we give here we extend this analogy~\cite{Kitaev06, Bombin10, Brown13, Wootton15, Zheng15} further.

\begin{figure}
\includegraphics{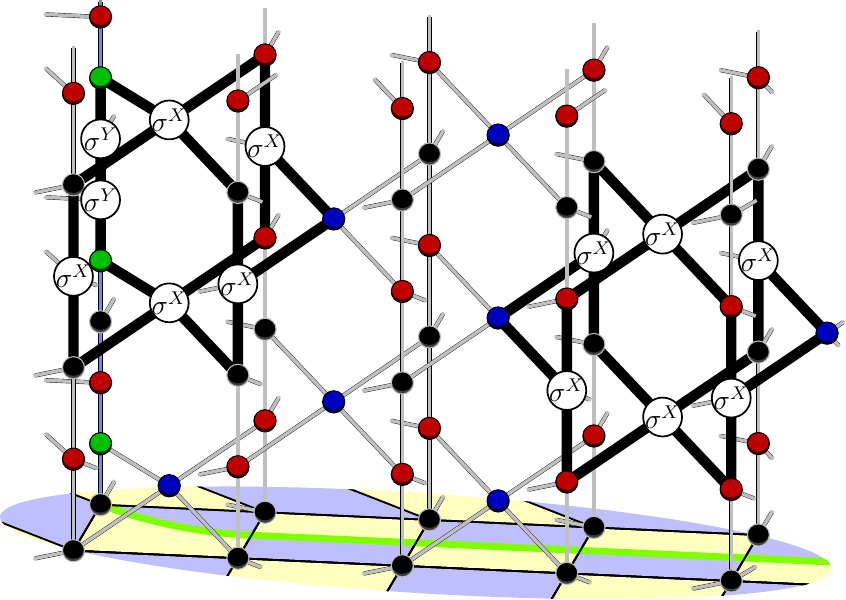}
\caption{\label{Fig:LiftedChecks} Resource state where we change the time interval of some of the targets of certain ancilla qubits as is described in the main text. Qubits indexed $X_j(t)$, $Y_j(t)$ and $Z_j(t)$ are colored red, green and black, respectively, and ancilla qubits are shown in blue. We also show the edges of the type-II foliated qubit in blue. No physical qubit has more than four incident edges. Two elements of $\mathcal{S} = \cent (\mathcal{F}) \cap \mathcal{R} \cap \mathcal{M}$ are shown.}
\end{figure}

We see that, in the spacetime picture provided by a foliation, that the corners of the planar code extend to world lines of Majorana modes. We find that the worldlines of the Majorana modes that live at the interface of different boundaries follow the trajectories we would expect if we were to realize fault-tolerant topological quantum computation by braiding Ising anyons~\cite{Moore91, Kitaev06, Brennen07, Nayak08, PreskillsLecture, Pachos12, Lahtinen17}. We show this analogy in a macroscopic picture in Fig.~\ref{Fig:Id} where we show the foliated planar code~\cite{Dennis02, Bolt16}. The planar code has rough boundaries and smooth boundaries, that appear as faces of the foliated system as they are extended along the time axis. The figure shows the different boundaries in blue and yellow. At the interface of the different boundary types we see the world line of a Majorana mode. In this picture the modes move vertically upwards with no horizontal motion. This is because with this channel we execute an identity gate. In what follows we show that generalizing the picture of foliation allows us to braid and fuse Majorana modes in the spacetime model. We see this by explicitly considering the initialization of arbitrary states, by showing lattice surgery by foliation, and by performing a fault-tolerant phase gate.

\subsection{initialization}

In order to realize fault-tolerant universal quantum computation with the surface code we require the ability to generate eigenstates of non-Pauli matrices for magic state distillation~\cite{Bravyi05}. This is similarly true for computation with the fault-tolerant cluster state model~\cite{Raussendorf07, Raussendorf07a}. State initialization by measurements with the surface code has been considered in Refs.~\cite{Landahl14, Lodyga15, Li15}. The work on initialization is readily adapted for the input-output model of fault-tolerant measurement-based quantum computation we have developed above.

Common to all of the references on initialization, the qubits that go on to form the surface code begin in some easily prepared state, such as a product state, which are then measured with the stabilizers of the desired code. In Fig.~\ref{Fig:initialization}(a) we show the initial state we prepare in order to initialize the surface code in an arbitrary state as given in Ref.~\cite{Lodyga15}. Qubits in the green(blue) boxes are initialized in known eigenstates of the Pauli-X(Pauli-Z) basis, and the central red qubit is prepared in an arbitrary state. 

\begin{figure}
\includegraphics{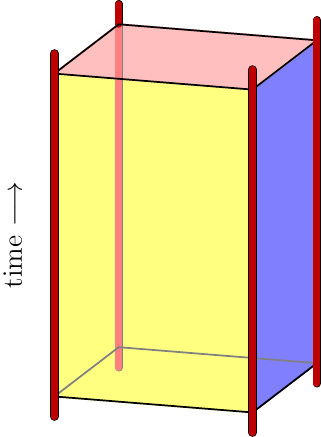}
\caption{\label{Fig:Id} The foliated planar code. The code is extended through spacetime where the rough and smooth boundaries are shown in yellow and blue, and the output is shown by the light red face. At the interface between the rough and smooth boundaries we see the world lines of Majorana modes, depicted by thick red lines. The four modes that encode the qubit move directly along the temporal axis, and as such do not nontrivially manipulate the encoded state over the channel.}
\end{figure}

By including this state as the stabilizer group $\Gin$ of the initial state $\mathcal{I}$, and the stabilizers of the surface code as the stabilizers measured through the channel, $\Gch$, we recover the graph state shown in Fig.~\ref{Fig:initialization}(b). The figure shows the input state on the layer closest to the reader and the time axis of the foliated model extends into the page. The qubits are colored using the convention in earlier sections where black(red) vertices mark qubits with indices $Z_j(t)$($X_j(t)$) and the blue vertices show ancilla qubits. The central green qubit is prepared in an arbitrary state. All of the qubits are measured in the Pauli-X basis\footnote{One could alternatively prepare the green qubit in a known eigenstate of Pauli-X and measure the green qubit in an arbitrary basis, the output will be the same.}. To help the qubits of the input system stand out we color the edges that connect them in light blue.

\begin{figure}
\includegraphics{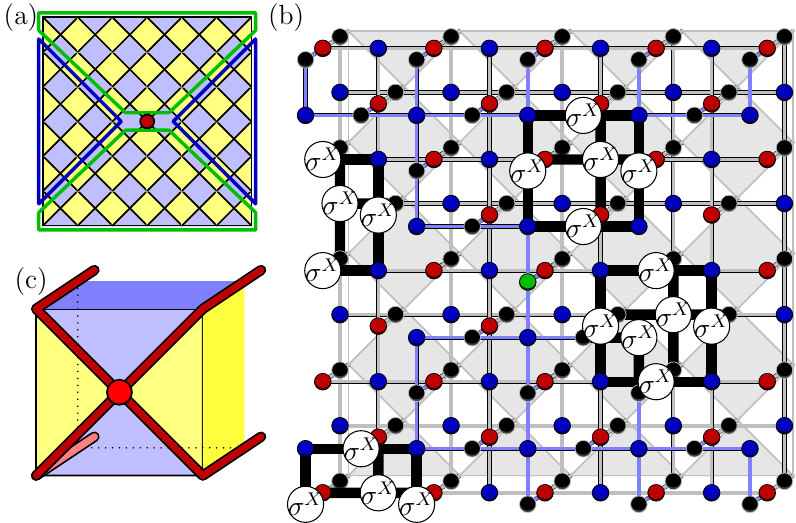}
\caption{\label{Fig:initialization} The input to encode an arbitrary state. (a)~The initial product state used to initialize the surface code from Ref.~\cite{Lodyga15}. Qubits lie on the vertices of the graph, and star(plaquette) operators lie on the blue(yellow) faces of the lattice. The system is initialized in a product state where all the qubits within the dark green(blue) boxes begin in the $+1$ eigenvalue eigenstate of the Pauli-X(Pauli-Z) operator, and the central red qubit is prepared in an arbitrary state. Measuring the stabilizers shown on the lattice initialize the code in the state of the central qubit. (b)~The graph state showing initialization where the initial state shown in (a) lies at the front of the figure, and the time axis of the foliated system extends into the page. Qubits initialized in an eigenstate of the Pauli-Z operator are not drawn, as these qubits do not entangle with the other qubits of the resource state. Examples of check operators (elements of $\cent (\mathcal{F}) \cap \mathcal{R} \cap \mathcal{M}$) are shown on the graph. (c)~initialization with the fault-tolerant cluster state shown macroscopically. Two pairs of Majorana modes are created at the site where the encoded state is initialized and their worldlines extend to the outer edges of the lattice in between the different boundaries which are determined by how the qubits of the initial system are prepared.}
\end{figure}

All of the qubits support check operators except the central qubit that is prepared in an initial state. We show examples of elements of $\cent(\mathcal{G})$ in the figure. We do not draw qubits that were initialized in an eigenstate of the Pauli-Z operator. Indeed, these states respond trivially to the action of the operator $U$ which is diagonal in the computational basis, and as such these qubits do not entangle nontrivially with the resource state. We are therefore free to neglect them.

It is interesting to examine initialization from the macroscopic viewpoint. It is well understood~\cite{Dennis02} that the surface code has two distinct boundaries, called rough and smooth boundaries, where the string-like logical Pauli-Z and Pauli-X operators, respectively, terminate. The initial state shown in Fig.~\ref{Fig:initialization}(a) is chosen such that both a logical Pauli-X and a logical Pauli-Z operator are supported on the physical qubits initialized in the Pauli-X(Pauli-Z) basis, where both operators intersect at the central qubit. Both of these logical operators extend along the surface where the input state is initialized onto the distinct boundaries where the logical operators terminate.

The respective boundaries of the foliated surface code extend along the temporal axis. Moreover, in the foliated system, the one-dimensional logical operators of the surface code are  extended along two-dimensional surfaces through the time axis onto the output system. The two-dimensional logical operators propagated along the time axis in this way are commonly known as correlation surfaces~\cite{Raussendorf06}. In the same way that the different logical operators of the surface code terminate at their distinct respective boundaries, so too do the correlation surfaces of the foliated surface code terminate exclusively at their respective boundary. The different boundaries of the foliated surface code, which is equivalent to the topological cluster state model~\cite{Raussendorf06}, are commonly known as primal or dual boundaries.

The boundary theory for the foliated surface code also extends to the side of the lattice where the system is initialized, see Fig.~\ref{Fig:initialization}(c). The primal and dual boundaries of the foliated surface code are distinct in the sense that they will only terminate the corresponding correlation surfaces of either the Pauli-X or Pauli-Z logical operator of the surface code. Likewise, when we initialize the surface code, the physical qubits are prepared in either the Pauli-X or Pauli-Z state depending on whether that area of the lattice should support either a logical Pauli-X or Pauli-Z operator upon initialization. In the foliated system, this choice of input state determines the type of correlation surface that terminates at this boundary. As such, we can consider an extension of the primal and dual boundaries on the side of the lattice where the system is initialized, that depends on how we initialized the qubits of the input surface code system. In Fig.~\ref{Fig:initialization}(c) we color code the boundaries of the foliated lattice according to the correlation surfaces that can terminate at them. In particular, correlation surfaces corresponding to the propagation of logical Pauli-X(Pauli-Z) operators of the surface code terminate at the blue(yellow) boundaries of the figure, respectively.

As discussed, we can regard the interface between the two different boundary types as worldlines of the Majorana modes. As Fig.~\ref{Fig:initialization}(c) shows, the four red lines meet at the single point where the nontrivial state is initialized. We thus see that the analogy between the boundaries of the foliated system and Majorana-based fault-tolerant quantum computation holds, as in such a system, to prepare an arbitrary state, two pairs of Majorana modes would need to be prepared simultaneously at a common location and noisily rotated into a desired state before the modes are separated such that the encoded information is topologically protected~\cite{Bravyi06, Karzig16}.

\begin{figure}
\includegraphics{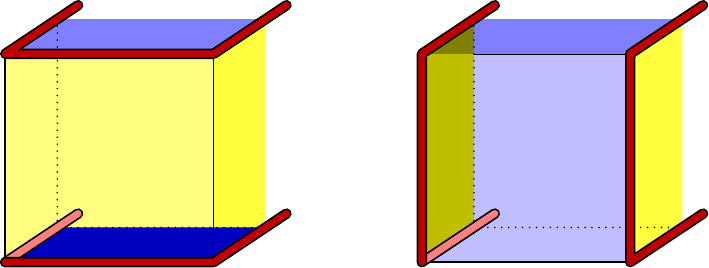}
\caption{\label{Fig:ZeroPlus} Macroscopic picture showing the fault-tolerant initialization of eigenstates of the logical Pauli operators. The time direction runs into the page. Initialization of the logical Pauli-Z (Pauli-X) operator is shown on the left-hand side (right-hand side) of the figure. The input state is prepared with the physical qubits initialized in known eigenstates of the Pauli-Z (Pauli-X) operator. The world lines of the Majorana modes that move between the distinct boundaries emerge in two well separated pair creation operations. This is true to the analogy with fault-tolerant quantum computation with Ising anyons.}
\end{figure}

We can similarly explore this analogy by preparing logical qubits in eigenstates of Pauli operators. With the surface code, we can fault-tolerantly prepare a logical qubit in an eigenstate of the Pauli-X(Pauli-Z) by initializing all of the physical qubits in the Pauli-X(Pauli-Z) basis. Using these product states as the input states, we can show the boundaries of the system macroscopically in Fig.~\ref{Fig:ZeroPlus} where the input state is shown by the boundary on the input face shown closest to the reader. In this instance, the figure shows two pairs of Majorana modes prepared a macroscopic distance from one another. This is similarly true with fault-tolerant quantum computation with Ising anyons, where these logical states are also prepared robustly in this way.

\subsection{Fault-tolerant parity measurements}

We next investigate the how lattice surgery~\cite{Horsman12, Landahl14, Brown17, Nautrup17} maps into the measurement-based picture. Lattice surgery offers a route to performing entangling gates between qubits encoded with topological codes via fault-tolerant parity measurements. Foliated lattice surgery with the surface code has already been considered in Ref.~\cite{Herr17b} but here we revisit this example within the more general framework we have developed for fault-tolerant measurement-based quantum computation. We will also witness nontrivial dynamics between the world lines of the Majorana modes that are present in the foliated surface code. We also remark that our discussion is readily extended to generalized lattice surgery where high-weight logical Pauli measurements are considered~\cite{Litinski19}.

We briefly review lattice surgery with the surface code before examining it at the macroscopic level of foliation. To the left and right of Fig.~\ref{Fig:LatticeSurgeryLattice} we see two surface code lattices in bold colors. Each of these encodes a single logical qubit. The goal is to make the logical parity measurement $\overline{Z}_1\overline{Z}_2$ of these two encoded qubits. The support of the operator we aim to measure is shown in the figure. It is important that the logical measurement is tolerant to local faults. To achieve this, we begin measuring the stabilizers of a single extended rectangular surface code, which includes the pale-colored stabilizers in between the two codes. Importantly, the logical operator of the two codes we intended to measure is a member of the stabilizer group of the extended code. As such, by measuring the stabilizers of the code we additionally recover the value of the parity measurement. 

\begin{figure}
\includegraphics{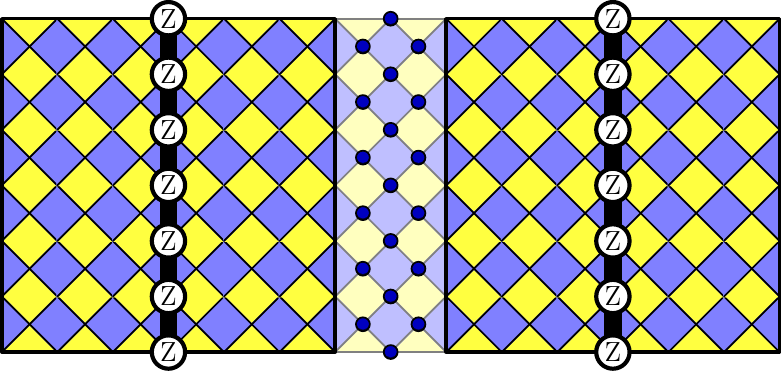}
\caption{\label{Fig:LatticeSurgeryLattice} Lattice surgery between two surface codes. Two surface code lattices are shown in bold colors to the left and the right of the figure. To measure the parity of these two logical qubits, we initialize the qubits that separate the two lattices in known eigenstates of the Pauli-X operator and then begin measuring the stabilizers of the extended rectangular lattice including the pale-colored stabilizers in between the two lattices.}
\end{figure}

To initialize the extended surface code, each of the physical qubits that lie in between the two encoded qubits are prepared in a known eigenstate of the Pauli-X operator and we then measure the stabilizers of the rectangular code. To the best of our knowledge, the literature thus far has only considered a very narrow separation between the two encoded lattices such that the number of qubits involved in the procedure is minimal. We know of no practical advantage of widening this gap but for the purposes of our exposition we find that the wider gap helps elucidate some of the topological features of lattice surgery.

We now consider lattice surgery within a measurement-based framework. We consider an input state $\Gin$ included in $\mathcal{I}$ with two surface code lattices together with some ancillary physical qubits prepared in the $+1$ eigenvalue eigenstate of the Pauli-X matrix as shown in bold in Fig.~\ref{Fig:LatticeSurgeryLattice}, where we might imagine that the two surface code lattices have emerged as the output of two foliated surface codes. We then append to the resource state the check measurements of the larger rectangular surface code such that the logical parity measurement is read from the resource state once the single-qubit measurement outcomes have been collected.

We show this foliated picture macroscopically in Fig.~\ref{Fig:Surgery}. The figure shows two foliated surface codes entering the surgery channel where the two lattices are connected to make the parity check. Once the parity check is completed, we take as an input the output of the channel of the rectangular code and input it back into the original channel where the two lattices are separated and the qubits that connect the two lattices are measured transversally in the Pauli-X basis, completing the operation.

We next study the trajectories of the Majorana modes while the parity measurement takes place. Given that the qubits that form the connection between the two codes are initialized in the Pauli-X basis, the boundary where the connection is formed terminates the correlation surface that propagates the logical Pauli-X data of the channel. As such, we regard this boundary as an extension of the blue boundaries in the figure where the correlation surfaces for the logical Pauli-X operators can terminate. An equivalent argument holds at the moment the connection is broken. From this we can infer the trajectories of the Majorana modes.

At the point the two codes are connected we observe the world lines of the Majorana modes that mark the interface between the distinct boundaries fuse. The product of these two fusion operations gives the parity of the two encoded qubits which is consistent with the topological interpretation of lattice surgery given in Ref.~\cite{Brown17}. We further remark that neither of the individual fusion measurements between pairs of Majorana modes reveal any of the logical information of the system.

\begin{figure}
\includegraphics{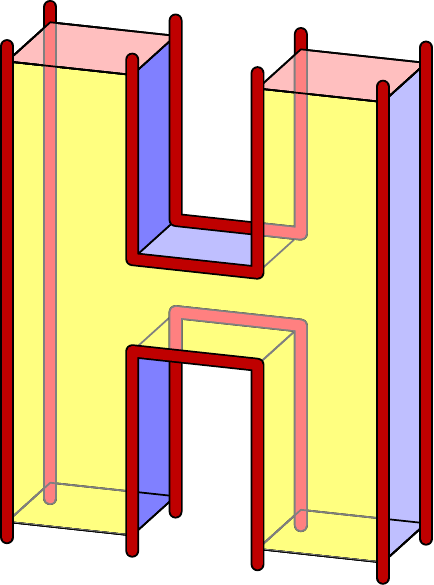}
\caption{\label{Fig:Surgery} Lattice surgery for the parity measurement $\overline{Z}_1 \overline{Z}_2$ with the foliated surface code. The time axis runs from the bottom of the page to the top. We show two distinct surface codes at the bottom of the page. At the point where we begin the parity measurement the two foliated qubits are connected through an extended piece of foliated surface code. We show the trajectories of the foliated Majorana modes in red. We see that at the connection point two pairs of Majorana modes are fused, where one mode from each pair are taken from the two different foliated codes. Once the parity data is collected the connection is broken by inputing the rectangular lattice into a new channel to separate the two encoded qubits. This is shown at the top of the figure.}
\end{figure}

More precisely, if we consider a system where two qubits are encoded over eight Majorana modes, $\gamma_j$ with $1 \le j \le 8$, and $\overline{X}_1 =\text{i} \gamma_1 \gamma_3$, $\overline{Z}_1 =\text{i} \gamma_3 \gamma_4$, $\overline{X}_2 =\text{i} \gamma_6 \gamma_8$ and $\overline{Z}_2 =\text{i} \gamma_5 \gamma_6$, and we make measurements $M_1 = \text{i} \gamma_3 \gamma_5$ and $M_2 = \text{i} \gamma_4 \gamma_6$ we have that $\overline{Z}_1 \overline{Z}_2 = M_1 M_2$. One may worry that the logical measurement of $M_1$ and $M_2$ may affect encoded information, but in fact this measurement only disturbs the global charge conservation of the two encoded qubits, $\gamma_1\gamma_2\gamma_3\gamma_4$ and $\gamma_5\gamma_6\gamma_7\gamma_8$. These operators could be regarded as gauge degrees of freedom. It is clear from Fig.~\ref{Fig:Surgery} that lattice surgery is performing an analogous operation with the foliated Majorana modes.

\subsection{A phase gate}
We finally show how to perform a phase gate with a surface code that is propagated through a resource state. This presents an interesting example as we require the composition of several channels to complete this operation. Moreover, we will observe a braid in the trajectories of the foliated Majorana modes which is true to the analogy we have painted alongside fault-tolerant quantum computation with anyons. The gate we use is based on a method presented with stabilizer codes in Ref.~\cite{Yoder17}, but we point out that the general theory of foliated quantum computation is readily adapted to other proposals to realize Clifford gates~\cite{Hastings15, Brown17} including other schemes presented in Ref.~\cite{Yoder17}. See also recent work in Ref.~\cite{Landahl14, Kesselring18}. As in the previous Subsection we will present the scheme at the level of stabilizer codes before discussing the foliated variant of the logical gate.

We first summarise the execution of a phase gate abstractly at the logical level. A phase gate maps logical operators such that $\overline{X} \rightarrow \overline{Y}$ and $\overline{Y} \rightarrow \overline{X}$ where phases are neglected, and the logical Pauli-Z operator is invariant under a phase rotation. Using an additional ancillary qubit we can achieve this operation by code deformation. If we encode the logical information on the first logical qubit, and we prepare the second qubit in an eigenstate of the logical Pauli-Z operator, then one can check that performing the following sequence of measurements; $\overline{Y}_1\overline{X}_2$, $\overline{Z}_1$, $\overline{X}_1\overline{X}_2$, $\overline{Z}_2$ will complete a phase gate up to phases which are determined by the outcomes of the measurements. In what follows we show how to make these fault-tolerant logical measurements using two surface codes. 

\begin{figure}
\includegraphics{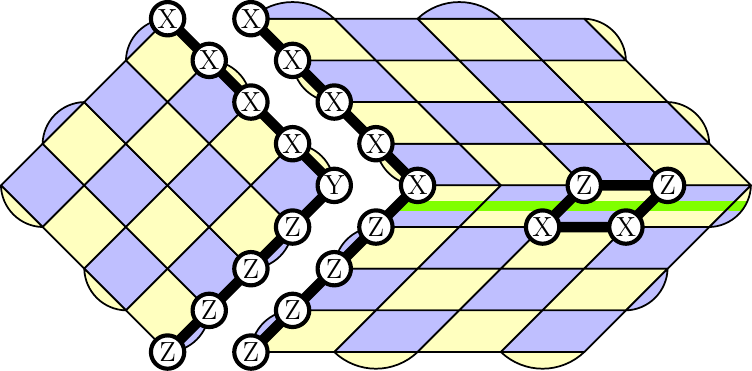}
\caption{\label{Fig:SC1} Two surface codes where Pauli-X type stabilizers lie on blue faces and Pauli-Z stabilizers on yellow faces. A dislocation is shown in green running through the middle of the right surface code which are the product of two Pauli-X terms and two Pauli-Z terms. An example of such a stabilizer is shown. The left qubit, qubit 1, is initially encoded in an arbitrary state, and the right qubit, qubit 2, is prepared in an eigenstate of the logical Pauli-Z operator. The logical operator whose measurement outcome is inferred from the stabilizer measurements during lattice surgery is shown explicitly.}
\end{figure}

We must first measure $\overline{Y}_1 \overline{X}_2$ where the ancilla qubit is prepared in an eigenstate of the Pauli-Z operator. We consider the initial system shown in Fig.~\ref{Fig:SC1} where logical information is encoded on the lattice shown to the left and the surface code to the right is initialized in an eigenstate of the logical Pauli-Z operator. An important feature of the surface code at the right of the figure is that it has a continuous defect line running through the middle of the lattice, but we remark that the model is locally equivalent to the well-known CSS variant of the surface code on a rectangular lattice. Once the system is prepared in this state, by measuring this system with the stabilizers of the twisted surface code as shown in Fig.~\ref{Fig:Twisted}, we recover the value of the desired logical parity measurement. This is because the operator of the initial system, $\overline{Y}_1 \overline{X}_2$, is an element of the stabilizer group of the twisted surface code. To make this clear we show the operator $\overline{Y}_1 \overline{X}_2$ explicitly on Fig.~\ref{Fig:SC1}.

It is worth pointing out that, similar to the standard surface code which, as explained above, can be fault-tolerantly initialized in an eigenstate of the Pauli-Z basis, the lattice to the right is readily initialized in an eigenstate of the Pauli-Z operator by initializing the physical qubits above(below) the defect line in a known eigenstate of the Pauli-Z(Pauli-X) basis before measuring the stabilizers of the model to complete the preparation.

Once we have performed the first parity measurement the remaining steps to complete the phase gate have already been well described in the literature, we thus only briefly summarise the remaining technical steps. We must first measure the first system in the logical Pauli-Z basis. This is achieved by measuring all of the physical qubits transversally in either the Pauli-Y or Pauli-Z basis, where the bulk of the system is measured in the Pauli-Z basis, with the exception of a single line of qubits that are measured in the Pauli-Y basis. The Pauli-Y measurements move the central twist shown in Fig.~9 to the bottom left-hand corner of second rectangular surface code as is explained in Ref.~\cite{Brown17}. The outcome of the measurement can be inferred from the single-qubit measurement outcomes. This leaves the logical information encoded on the second logical qubit on the rectangular lattice.

To make the final parity measurement, we additionally require that we reduce the length of the rectangular lattice. This is also achieved by transversally measuring the qubits below the defect line on the rectangular lattice transversally in the Pauli-Z basis. After we have completed all of the transversal measurements we reinitialize the first system in an eigenstate of the logical Pauli-Z operator which is carried out by preparing all of the physical qubits in the Pauli-Z basis and subsequently measuring standard surface code stabilizers.

\begin{figure}
\includegraphics{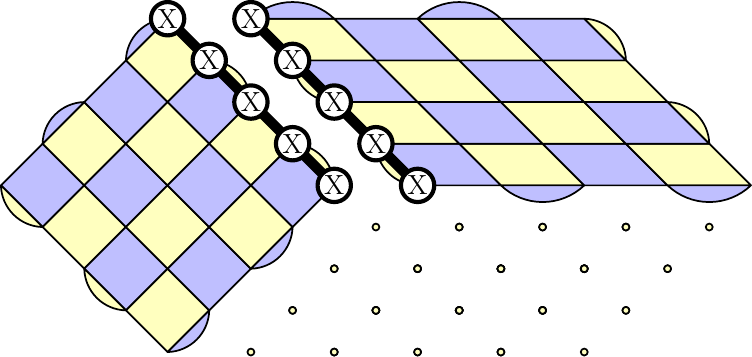}
\caption{\label{Fig:SC3} After transversally measuring the qubits below the defect line of the lattice to the right transversally in the computational basis the right qubit is reduced to a smaller square lattice. The system to the left is initialized in a logical eigenstate of the Pauli-Z operator. The support of the logical parity measurement that is made under the surgery is shown.}
\end{figure}

Upon completing these three operations, the first two of which could be carried out simultaneously, we end with the system shown in Fig.~\ref{Fig:SC3}. Finally, to transfer the rotated logical information back to the first lattice, we perform another logical parity measurement, $\overline{X}_1 \overline{X}_2$ which is carried out using standard lattice surgery that we discussed in the previous subsection before finally measuring the remaining qubits of the second ancillary system transversally in the Pauli-Z basis which completes the final logical measurement $\overline{Z}_2$.

As in the case of lattice surgery, the steps of the code deformation procedure outlined above are readily mapped onto a foliated system by using each new deformation as the stabilizers of the channel system of a resource state such that the output is the deformed variant of the input state. The phase gate we have presented provides a particularly interesting example as it shows that certain gates are achieved using several different channels. We also remark that several of these channels require the general methods of code foliation that we have developed in the earlier sections of this work. The channel where we measure the stabilizers of the twisted surface code for instance can be foliated using the graph states proposed in Sec.~\ref{Sec:TwistedSurfaceCode}.

\begin{figure}
\includegraphics{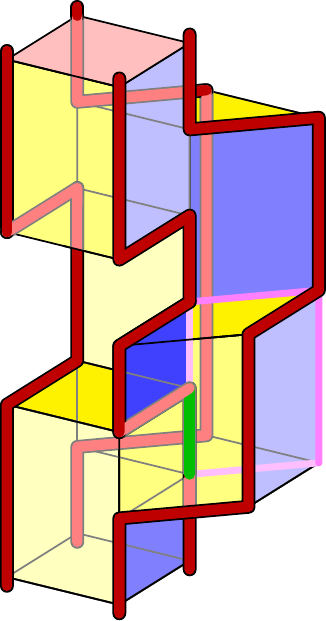}
\caption{\label{Fig:TwistExchange} The sequence of foliated channels that execute a phase gate on a qubit transmitted with a foliated surface code. Time increases up the page. The worldlines of the Majorana modes are shown in red(green) if they run through the exterior boundary(bulk) of the foliated system and the boundary of the dislocation of the foliated surface code is shown in pink. One can track the world lines of the modes from the bottom to the top of the page to see the two right most world lines exchange.}
\end{figure}

We show the series of channels in Fig.~\ref{Fig:TwistExchange} where the temporal axis increases up the page. Interestingly, we observe that the two right-most Majorana modes at the bottom of the figure are exchanged over time, where the right-most mode moves through the interior of the resource state through the channel that measures the stabilizers of the twisted surface code. The point where the world line of this mode moves through the interior is highlighted in green. To the best of our knowledge, this presents the first example of a foliated system where a defect is moved through the interior of a resource state. While this defect is in the interior of the system, the other mode involved in the exchange is braided around the exterior of the system before the exchange is completed thus executing a phase gate. Once again, the analogy with Majorana modes holds in the model we consider here as, indeed, exchanging a pair of Majorana modes executes phase gate. This can be seen by considering a qubit encoded with four Majorana modes such that $\overline{X} = \text{i}\gamma_1 \gamma_3$, $\overline{Y} = \text{i} \gamma_2 \gamma_3$ and $\overline{Z} = \text{i} \gamma_1 \gamma_2$. One can see that up to phases $\overline{X}$ and $\overline{Y}$ differ by the exchange of indices $1$ and $2$ whereas $\overline{Z}$ is invariant under the exchange modulo a negative phase. The operation is equivalent to the exchange of the modes in the foliated channel.

\section{Concluding remarks}
We have presented a framework that allows us to map quantum computational schemes that use code deformations on stabilizer codes to schemes of fault-tolerant measurement-based quantum computation. As an example, we have used this framework to show we can initialize, fuse and braid foliated Majorana defects in a three-dimensional fault-tolerant cluster-state model to carry out universal quantum computation with Clifford operations and magic state distillation.

It remains an important problem to find the most resource efficient models of fault-tolerant quantum computation. As such it will be fruitful to study the robustness of other foliated stabilizer codes to experimentally relevant sources of noise. In particular, it will be valuable to study the tolerance of different models to loss, as foliated models are most applicable to photonic architectures where this is a dominant source of error~\cite{Rudolph16}. One may also be able to adapt known results using the circuit-based model to find fault-tolerance thresholds for the measurement-based methods of fault-tolerant quantum computation presented here~\cite{Aliferis05}. Recent work in Refs.~\cite{Bolt18, Brown19} may offer a way to find decoding algorithms for generic foliated codes that will be necessary to determine fault-tolerant threshold error rates.

From a condensed matter perspective, the models we have constructed can be viewed as symmetry-protected topological phases~\cite{TrithepSPT, roberts2018symmetry,YoshidaCCSPT,YoshidaHigher,ThermalSPT,raussendorf2018computationally, roberts2019symmetry}. Further study of these models may therefore lead to new phases of matter that are robust resources for quantum information processing tasks. Finally, one can readily check that we can foliate the canonical examples of subsystem codes~\cite{Bacon06, Bombin10a, Bravyi13, Bombin15} within our framework. We check this by replacing the channel stabilizer group by a generating set of the gauge group, such that the foliated system inherits many of the desirable features of the subsystem code. We give a discussion on this in Appendix~\ref{App:SubsystemFoliation}. While the general theory of subsystem code foliation remains to be described explicitly, we find it exciting to map the advantageous characteristics of these models such as gauge fixing~\cite{Paetznick13, Bombin15, Yoder17a} and single-shot error-correction~\cite{Bombin15a, Brown16a, Campbell18,fawzi2018constant} into foliated systems in the future.

\begin{acknowledgements}
The authors are grateful to H. Bomb\'{i}n for discussions on the connection between code deformation and gauge fixing, to N. Nickerson for conversations about foliation and the topological cluster state model at the earlier stages of this project, to M. Gimeno-Segovia for directing us through the photonic quantum computing literature, and to M. Newman and K. Seetharam for comments on our work. The authors are supported by the Australian Research Council via the Centre of Excellence in Engineered Quantum Systems(EQUS) project number CE170100009. BJB acknowledges support from the University of Sydney Fellowship Programme. SR is also supported by the Australian Institute for Nanoscale Science and Technology Postgraduate Scholarship (John Makepeace Bennett Gift).
\end{acknowledgements}

\appendix

\section{Type-II foliated qubits}
\label{App:TypeII}

In this Appendix we consider a variation of the foliated qubit where we perform single-qubit Pauli-Y measurements instead of Pauli-X measurements to move information along a one-dimensional cluster state. We refer to these as type-II foliated qubits. Qubit foliation by this method offers a natural way to measure the encoded Pauli-Y information by coupling an ancilla to just one qubit of the chain. This may be of practical benefit as it can reduce the valency of the graph of a resource state at the expense of including three qubits in each time interval of the foliated qubit instead of two in the case of a type-I folated qubit. 

One can produce a foliated system $\mathcal{F}$ based on type-II qubits by following the prescription in Sec.~\ref{Sec:CodeFoliation} where some or all of the foliated qubits in the channel $\mathcal{K}$ are replaced by type-II foliated qubits, and then producing a new resource $\mathcal{R}$ and measurement pattern $\mathcal{M}$ to perform the required parity measurements. Theorem~\ref{Thm:OutputCode} and Theorem~\ref{Thm:ChannelStabilizers} are readily modified to accommodate this change.

We begin with the cluster state of length $N = 3D+1$ with stabilizer group defined in Eqns.~(\ref{Eqn:ClusterStab1}) and~(\ref{Eqn:ClusterStab2}) and logical operators defined in Eqn.~(\ref{Eqn:ClusterLogicals}) such that we can measure the logical Pauli-Z information with a single-qubit Pauli-Z measurement on the first qubit. 

We notice the difference between type-I and type-II foliated qubits by first looking at what happens if we make a Pauli-Y measurement instead of a Pauli-X measurement on the first qubit. We first multiply both logical operators by the stabilizer $C[2] = \sigmaz{1} \sigmax{2} \sigmaz{3}$ such that both logical operators commute with $M_1 = \sigmay{1}$. We have
\begin{equation}
X \sim  \sigmay{1} \sigmay{2} \sigmaz{3}, \quad Z \sim \sigmax{2} \sigmaz{3},
\end{equation}
With these logical operators it is easily checked that measuring $M_1 = \sigmay{1}$, which becomes the stabilizer $C[1] = y_1 \sigmay{1}$, where $y_1 = \pm 1$ is the random measurement outcome, that after the measurement we have the logical operators
\begin{equation}
X \sim y_1 \sigmay{2} \sigmaz{3}, \quad Z \sim \sigmax{2} \sigmaz{3}.\label{Eqn:LogicalXZ}
\end{equation}
Notably, we also have that
\begin{equation}
Y= i X Z \sim y_1 \sigmaz{2},
\end{equation}
using the logical operator expressions given in Eqn.~(\ref{Eqn:LogicalXZ}). From this equation it is easily seen that logical Pauli-Y information can be accessed by making a single-qubit Pauli-Z measurement on the second qubit provided the first qubit is measured in the Pauli-Y basis. 

One can then check that measuring the second qubit in the Pauli-Y basis, whose measurement outcome is  $y_2 = \pm 1$, we obtain
\begin{equation}
X \sim y_1 y_2 \sigmaz{3}, \quad Z \sim y_2 \sigmay{3} \sigmaz{4},
\end{equation}
and in turn $ Y \sim y_1 \sigmax{3} \sigmaz{4}$. With the example of measuring the first two qubits along the cluster state in the Pauli-Y basis, we see that we cyclicly permute the logical operator that can be accessed with a single-qubit Pauli-Z measurement as we progress along the chain. In contrast, the type-I foliated qubit exchanges the information that is accessible by single-qubit measurements between logical Pauli-X and Pauli-Z data.

\begin{figure}
\includegraphics{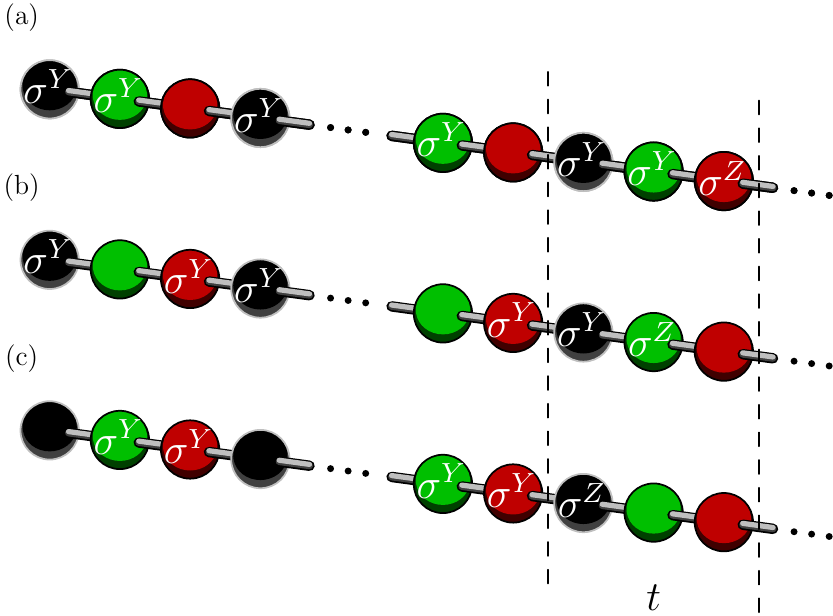}
\caption{\label{Fig:TypeIIinterval} We show the logical Pauli-X, Pauli-Y and Pauli-Z operators in (a), (b)~and (c)~respectively. From these diagrams we see that by measuring the appropriate qubit within the time interval in the Pauli-Z basis, and otherwise measuring the preceeding qubits of the system in the Pauli-Y basis, we can measure the logical information of the chain in an arbitrary Pauli basis.}
\end{figure}

As in the case of the type-I foliated qubit, it is convenient to redefine the indices of the system in terms of intervals, indexed by $t$. For type-II foliated qubits we define
 \begin{equation}
 X(t) = 3t ,\quad Y(t) = 3t-1,\quad Z(t) = 3t - 2
 \end{equation}
In this case we have intervals of three adjacent qubits where either the logical Pauli-X, Pauli-Y, or Pauli-Z measurement information can be accessed via single qubit measurements. Specifically, we have that
\begin{equation}
 X \sim \Sigma^X(t) \sigmaz{X(t)},  \quad Y \sim \Sigma^Y(t) \sigmaz{Y(t)} ,
\end{equation}
and
\begin{equation}
 Z \sim \Sigma^Z(t) \sigmaz{Z(t)},
\end{equation}
where we have defined
 \begin{equation}
 \Sigma^X(t) = \prod_{\mu = 1}^{t} \sigmay{Y(\mu)} \sigmay{Z(\mu)} ,\end{equation} 
 \begin{equation}
 \Sigma^Y(t) = \sigmay{Z(t)} \prod_{\mu = 1}^{t-1} \sigmay{X(\mu)} \sigmay{Z(\mu)},
\end{equation}
 and 
\begin{equation} 
\Sigma^Z(t) =  \prod_{\mu = 1}^{t-1} \sigmay{X(\mu)} \sigmay{Y(\mu)}.
\end{equation}
which commute with the measurement pattern of type-II foliated qubits. We show a time interval for a type-II foliated qubit in Fig.~\ref{Fig:TypeIIinterval}.

\subsection{Ancilla-assisted measurement}

We can also perform measurements on the logical qubit propagated along a type-II foliated qubit. In Fig.~\ref{Fig:NonDestructiveMeasurement} we show the operator $\Sigma^Y(t)\sigmax{a} \in \mathcal{M}$ where we have the measurement pattern of a type-II foliated qubit, namely,
\begin{equation}
\mathcal{M}^C = \left\langle  \left\{ \sigmay{\mu}\right\} _{\mu=1}^{N-1}\right\rangle
\end{equation}
and the blue ancilla qubit which is measured in the Pauli-X basis is coupled to target $T = Y(t)$.

Of course, type-II foliated qubits can support parity measurements in a foliated system of multiple foliated qubits using the methods specified above by means of an ancilla. One caveat of which is that we must also modify the measurement basis of the ancilla accordingly. Indeed, we previously considered measuring $M[a] \in \mathcal{M}^A$ where we have thus far considered the case where $M[a] = \sigmax{a} \sigmaz{a}^{p^X\cdot p^Z}$ to measure $P$ such that $p = v(P)$. We adapt this measurement from $\sigmax{a}$ to $\sigmay{a}$ or vice versa for each type-II qubit where we replace the coupling from ancilla $a$ to both $X_j(t)$ and $Z_j(t)$ with a single $Y_j(t)$ coupling.

\begin{figure}
	\includegraphics{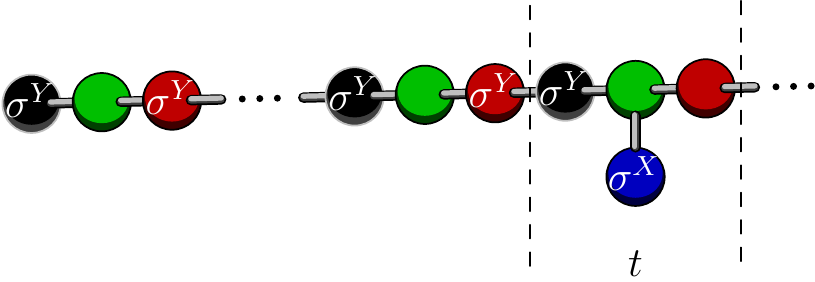}
	\caption{\label{Fig:NonDestructiveMeasurement} The graph for a type-II foliated qubit where we make a non-destructive logical Pauli-Y measurement at time interval $t$. The blue ancilla qubit is coupled to the chain element indexed $T = Y(t)$ via a controlled-phase gate. We show the logical Pauli-Y operator that commutes with the single-qubit measurement pattern wherethe ancilla is measured in the Pauli-X basis and otherwise, the qubits in the chain are measured in the Pauli-Y basis, as prescribed.}
\end{figure}

\section{Stabilizers and logical operations of the foliated system}\label{App:ThmProofs}

In this Appendix we prove Theorem~\ref{Thm:OutputCode} and~\ref{Thm:ChannelStabilizers} of Sec.~\ref{Sec:CodeFoliation}. We first define two Pauli operators that we use throughout this appendix and we prove two lemmas before moving onto the proofs of the theorems in the main text.

We will frequently consider the action of operators $R_P$ and $R'_P(t)$ on the resource state. We will be specifically interested in the commutation relations of these operators with the degrees of freedom introduced to $\mathcal{R}$ by the input code $\Gin$. We remeber that these degrees of freedom were initialized on qubits indexed $Z_j(1)$ on the initial state $\mathcal{I}$. They are defined as follows
\begin{eqnarray}
R_P &=& \prod_{G \in \Gres}  \prod_{t \le D}  \sigmaz{G(t)}^{\Upsilon( p,  g)} \label{Eqn:PinR}  \\
&& \times \prod_{j \in \supp{p^X}}\Sigma^X_j(D) \sigmax{Z_j(D+1)}  \nonumber  \\
&& \times \prod_{j \in \supp{p^Z}}\Sigma^Z_j(D+1)\sigmaz{Z_j(D+1)}, \nonumber 
\end{eqnarray}
and
\begin{eqnarray}
R'_P(t) &=& \prod_{G \in \Gres}  \left( \sigmaz{G(t)}^{p^X \cdot g^Z} \prod_{t' < t}  \sigmaz{G(t')}^{\Upsilon( p,  g)} \right) \nonumber \\
&& \times  \prod_{j \in \supp{p^X}}\Sigma^X_j(t) \sigmaz{X_j(t)} \label{Eqn:PinRt}  \\
&& \times  \prod_{j \in \supp{p^Z}}\Sigma^Z_j(t)\sigmaz{Z_j(t)}, \nonumber 
\end{eqnarray}
where $P$ is a Pauli operator with corresponding vector representation $(p^X \, p^Z)^T = v(P) $ and $(g^X \, g^Z) = v(G)$ for each $G \in \Gres$.

The relationship between operators $R_P$ and $R'_P(t)$ and degrees of freedom specified by the input code $\Gin$ are described by the following lemma
\begin{lem} \label{Lem:Channel}
Operators $R_P$ and $R_P'(t)$ act on code states of the resource state, $\mathcal{R}$, denoted $| \psi \rangle $, such that
\begin{equation}
R_P | \psi \rangle = R'_P(t)|\psi \rangle = U^A U^C I_P |\phi \rangle 
\end{equation}
for code states $|\phi \rangle = \left( {U^A}{U^C}\right)^\dagger |\psi \rangle$ of the initial stabilizer code $\mathcal{I} \otimes \mathcal{A}$ 
where
\begin{equation}
I_P = \prod_{j \in \left|p^X\right|} \sigmax{Z_j(1)} \prod_{j \in \left|p^Z\right|} \sigmaz{Z_j(1)}.  \label{Eqn:IinR}
\end{equation}

\end{lem}

\begin{proof}
The action of unitary operator $U^C$ on $I_P |\phi \rangle$ gives
$\left( U^C I_P {U^C}^\dagger \right) U^C |\phi \rangle = U^C I_P |\phi \rangle $
where, written explicitly, we have
\begin{equation}
U^C I_P {U^C}^\dagger = \! \prod_{j \in \left|p^X\right|} \! \sigmax{Z_j(1)}  \sigmaz{X_j(1)} \! \prod_{j \in \left|p^Z\right|} \! \sigmaz{Z_j(1)}, \label{Eqn:IntermediatePinK}
\end{equation}
and $U^C |\phi \rangle$ is a codestate of $\chan$. Multiplying by elements of the stabilizer group $\mathcal{K}$ we have $K_P ,\, K_P'(t) \sim U^C I_P {U^C}^\dagger $ such that
\begin{eqnarray}
K_P &=& \prod_{j \in \supp{p^X}} \Sigma^X_j(D)\sigmax{Z_j(D+1)}  \label{Eqn:PinK} \\ 
&& \times  \prod_{j \in \supp{p^Z}} \Sigma^Z_j(D+1)\sigmaz{Z_j(D+1)}. \nonumber 
\end{eqnarray}
and
\begin{equation}
K'_P(t) = \prod_{j \in \supp{p^X}} \Sigma^X_j(t)\sigmaz{X_j(t)}  \prod_{j \in \supp{p^Z}} \Sigma^Z_j(t)\sigmaz{Z_j(t)}. \label{Eqn:PinKt}
\end{equation}
By the definition of stabilizer operators, both $K_P$ and $K'_P(t)$ of Eqns.~(\ref{Eqn:PinK}) and~(\ref{Eqn:PinKt}) acts equivalently on the codespace of $\chan$ to $U^C I_P {U^C}^\dagger$ Eqn.~(\ref{Eqn:IntermediatePinK}).

We obtain Eqns.~(\ref{Eqn:PinR}) and~(\ref{Eqn:PinRt}) by conjugating $K_P$ and $K_P'(t)$ with $U^A$ such that $R_P = U^A K_P {U^A}^\dagger$ and $R_P'(t) = U^A K'_P(t) {U^A}^\dagger$. Lemma~\ref{Lem:Channel} holds by the unitarity of $U^A$.
\end{proof}

Broadly speaking, Lemma~\ref{Lem:Channel} shows how operators $P$ supported on qubits $Z_j(1)$ of the initial system are supported on $\mathcal{R}$ under conjugation by unitary $U^CU^A$. It is particularly useful for operators $P \in \cent(\Gch)$ whereby $\Upsilon(p,g) = 0$. Nevertheless, it will also be important to determine if $I_P \in \mathcal{I}$ is measured by the resource state, i.e., if there are operators $\sim R_P'(t)$ that are members of $\mathcal{M}$. As such, we must also consider elements of $\mathcal{R}$ that are obtained by conjugation of elements of the ancilla system $\mathcal{A}$ with entangling unitary $U^A$. This is characterized by Lemma~\ref{Lem:Checks} as follows

\begin{lem}
\label{Lem:Checks}
The stabilizer group of the resource state includes terms $ C[G(t)] \in \mathcal{R} $ for $G \in \Gres$ such that
\begin{eqnarray}
C[G(t)] &=& \sigmax{G(t)} \prod_{ \substack{\tilde{G}\in \Gres,\\ \tilde{G} \not= G}} \sigmaz{\tilde{G}(t)}^{g^X\cdot \tilde{g}^Z} \label{Eqn:ClusterTerm} \\
&& \times  \prod_{j \in |g^X|} \sigmaz{X_j(t)} \prod_{j \in |g^Z|} \sigmaz{Z_j(t)}, \nonumber
\end{eqnarray}
where $(g^X \, g^Z)^T = v(G)$ and $(\tilde{g}^X \, \tilde{g}^Z)^T = v(\tilde{G})$. 
\end{lem}

\begin{proof}
By Def.~\ref{Def:Resource} we have $\mathcal{R} = U^A (\chan \otimes \mathcal{A})$ and $ \sigmax{G(t)} \in \chan \otimes \mathcal{A}$ since $ \sigmax{G(t)} \in \mathcal{A}$. We obtain Eqn.~(\ref{Eqn:ClusterTerm}) using that $ C[G(t)]  = U^A \sigmax{G(t)} {U^A}^\dagger $ which follows from $U^A$ given explicitly in Eqn.~(\ref{Eqn:UA}) of Def.~\ref{Def:Resource}.
\end{proof}

We are also interested in the measurement outcomes of terms $ P \in \Gch $ of the input data that are not necessarily included $ \Gres $. It will be helpful to define the subset $\xi(P) \subseteq \Gres$ such that $ P = \prod_{G \in \xi(P)} G $. The subset $\xi(P)$ must exist for any $P \in \Gch$ by the definition of $\Gres$, i.e., $\Gres$ is a generating set of $\Gch$.

\begin{cor}
The term $C[\xi(P)(t)] \equiv \prod_{G \in \xi(P)} C[G(t)] \in \mathcal{R}$ is such that
\begin{eqnarray}
C[\xi(P)(t)] &=& \prod_{\tilde{G}\in \Gres}  \sigmaz{\tilde{G}(t)}^{ p^X \cdot \tilde{g}^Z } \label{Eqn:Cap}  \\
&& \times \prod_{G\in\xi(P)}\sigmax{G(t)} \sigmaz{G(t)}^{ g^X\cdot g^Z}  \nonumber \\
&& \times \prod_{j\in\supp{{p^X}}} \sigmaz{X_j(t)} \prod_{j\in\supp{{p^Z}}} \sigmaz{Z_j(t)}. \nonumber
\end{eqnarray}
where we write $\left({p^{X}} \, {p^Z} \right)^T = v\left(P \right)$.
\end{cor}
The above corollary is obtained as follows. We have 
\begin{eqnarray}
C[\xi(P)(t)] &=& \prod_{G\in\xi(P)}  \sigmax{G(t)}  \label{Eqn:CorollaryStep}  \\
&& \times 
\prod_{G\in\xi(P)} \left(
\prod_{\substack{\tilde{G}\in \Gres,\\ \tilde{G} \not= G}} \sigmaz{\tilde{G}(t)}^{g^X\cdot \tilde{g}^Z} \right)
\nonumber \\
&& \times \prod_{j\in\supp{{p^X}}} \sigmaz{X_j(t)} \prod_{j\in\supp{{p^Z}}} \sigmaz{Z_j(t)}, \nonumber
\end{eqnarray}
using that $C[\xi(P)(t)] = \prod_{G\in \xi(P)} C[G(t)] $ and the expression for $C[G(t)]$ given in Eqn.~(\ref{Eqn:ClusterTerm}) in Lemma~\ref{Lem:Checks}. We obtain the above expression for $C[\xi(P)(t)]$ in Eqn.~(\ref{Eqn:Cap}) then using the following identity together with Eqn.~(\ref{Eqn:CorollaryStep})
\begin{eqnarray}
&& \prod_{G\in\xi(P)} \left( \prod_{\substack{\tilde{G}\in \Gres,\\ \tilde{G} \not= G}} \sigmaz{\tilde{G}(t)}^{g^X\cdot \tilde{g}^Z} \right) \nonumber \\
&& =  \prod_{G\in\xi(P)} \left(  \sigmaz{G(t)}^{g^X\cdot g^Z} \prod_{\tilde{G}\in \Gres} \sigmaz{\tilde{G}(t)}^{g^X\cdot \tilde{g}^Z} \right) \nonumber \\
&& = \prod_{\tilde{G}\in \Gres}  \sigmaz{\tilde{G}(t)}^{ \tilde{g}^Z \cdot \sum_{G \in \xi(P)}g^X} \prod_{G\in\xi(P)} \sigmaz{G(t)}^{ g^X\cdot g^Z}  \nonumber \\
&& = \prod_{\tilde{G}\in \Gres}  \sigmaz{\tilde{G}(t)}^{ p^X \cdot \tilde{g}^Z } \prod_{G\in\xi(P)} \sigmaz{G(t)}^{ g^X\cdot g^Z} . \nonumber
\end{eqnarray}

We are now in a position to prove Theorem~\ref{Thm:OutputCode} and \ref{Thm:ChannelStabilizers}.

\begin{proof1}
We deal with Eqns.~(\ref{Eqn:OutputStabilizer}) and~(\ref{Eqn:OutputLogical}) separately before finally examining elements of $\cent(\Gin) \backslash \Gin$ that are measured by the foliated system. We begin by determining $\Gout$. 
\begin{proof1a}
For any element of $ P \in \Gin$ with $(p^X \, p^Z)^T = v(P)$ we have a stabilizer of the initial state $ I_P \in \mathcal{I} \otimes \mathcal{A} $ of the form Eqn.~(\ref{Eqn:IinR}) such that we have $R_P \in \mathcal{R}$ defined according to Eqn.~(\ref{Eqn:PinR}) by Lemma~\ref{Lem:Channel}. Then, provided $P \in \cent(\Gch)$ we have that $\Upsilon(p,g) = 0$ for all $g = v(G)$ with $G \in \Gres $. Therefore, we have
\begin{eqnarray}
R_P &=&  \prod_{j \in \supp{p^X}}\Sigma^X_j(D) \sigmax{Z_j(D+1)}    \\
&& \times \prod_{j \in \supp{p^Z}}\Sigma^Z_j(D+1)\sigmaz{Z_j(D+1)}, \nonumber
\end{eqnarray}
for $P \in \Gin \cap \cent(\Gch)$. Given that $\Sigma^X_j(D),\, \Sigma^Z_j(D+1) \in \mathcal{M}$ we have that 
\begin{equation}
\prod_{j \in \supp{p^X}} \sigmax{Z_j(D+1)}  \prod_{j \in \supp{p^Z}}\sigmaz{Z_j(D+1)} \in \Gout, \label{Eqn:GoutPT1}
\end{equation}
up to the measurement outcomes inferred from $\mathcal{M}$.

We next show elements $ P \in \Gch $ are elements of $\Gout$. We note that $ P^2 = 1 \in \Gin$. Likewise, $I_P^2 \in \mathcal{I} \otimes \mathcal{A}$ by unitarity of $U^C$. Therefore, by Lemma~\ref{Lem:Channel} we have  $R'_P(t)R_P \in \mathcal{R}$. We write this explicitly
\begin{eqnarray}
R'_P(t) R_P  &=&  \prod_{G \in \Gres}  \sigmaz{G(t)}^{p^X \cdot g^Z}  \label{Eqn:RdashR}  \\ 
&& \times  \prod_{j \in \supp{p^X}}  \sigmaz{X_j(t)}  \prod_{j \in \supp{p^Z}} \sigmaz{Z_j(t)}\nonumber \\
&& \times \prod_{j \in \supp{p^X}}  \Sigma^X_j(t) \Sigma^X_j(D) \sigmax{Z_j(D+1)}  \nonumber  \\
&& \times \prod_{j \in \supp{p^Z}} \Sigma^Z_j(t) \Sigma^Z_j(D+1)\sigmaz{Z_j(D+1)}, \nonumber 
\end{eqnarray}
where $\Upsilon(p,g) = 0$ since $\Gch$ is Abelian. Now, given that $P \in \Gch$ there exists some $\xi(P)$ such that we have $C[\xi(P)(t)] \in \mathcal{R}$ as shown in Eqn.~(\ref{Eqn:Cap}).

The product of Eqn.~(\ref{Eqn:Cap}) and~(\ref{Eqn:RdashR}) then gives
\begin{eqnarray}
C[\xi(P)(t)] R'_P(t) R_P  &=& \prod_{G \in \xi(P)}  \sigmax{G(t)} \sigmaz{G(t)}^{g^X \cdot g^Z} \nonumber  \\
&&  \times \prod_{j \in \supp{p^X}}  \Sigma^X_j(t) \Sigma^X_j(D)   \label{Eqn:CorrelationSurfaceDef} \\
&& \times  \prod_{j \in \supp{p^Z}} \Sigma^Z_j(t) \Sigma^Z_j(D+1) \nonumber \\
&&  \times \prod_{j \in \supp{p^X}}  \sigmax{Z_j(D+1)}  \nonumber  \\
&& \times  \prod_{j \in \supp{p^Z}} \sigmaz{Z_j(D+1)}. \nonumber 
\end{eqnarray}
Since $M[G(t)] = \sigmax{G(t)} \sigmaz{G(t)}^{g^X\cdot g^Z} \in \mathcal{M}^A$ and $\Sigma^X_j(t),\Sigma^Z_j(t) \in \mathcal{M}^C$ we find that their values are inferred from $\mathcal{M}$. We therefore find the term
\begin{equation}
\prod_{j \in \supp{p^X}} \sigmax{Z_j(D+1)}  \prod_{j \in \supp{p^Z}} \sigmaz{Z_j(D+1)} \in \Gout, \label{Eqn:GoutPT2}
\end{equation}
for $P \in \Gch$ up to the measurement outcomes of $\mathcal{M}$. The results of Eqn.~(\ref{Eqn:GoutPT1}) and Eqn.~(\ref{Eqn:GoutPT2}) show that $\mathcal{G}_{\text{ch.}} \cup [\mathcal{G}_{\text{in}} \cap \mathcal{C}(\mathcal{G}_{\text{ch}})] \subseteq \mathcal{G}_{\text{out}}$. A similar argument can be made to show $\mathcal{G}_{\text{out}} \subseteq \mathcal{G}_{\text{ch.}} \cup [\mathcal{G}_{\text{in}} \cap \mathcal{C}(\mathcal{G}_{\text{ch.}})]$, thus verifying Eq.~(\ref{Eqn:OutputStabilizer}).
\end{proof1a}

\begin{proof1b}
We now turn to the logical operators as determined by $\mathcal{L}_{\text{out}} = \cent(\Gout) \backslash \Gout$. We require that 
\begin{equation}
P \in \cent(\Gin) \backslash \Gin, \label{Eqn:LoutPT1}
\end{equation} 
such that $R_P \in \cent(\mathcal{R}) \backslash \mathcal{R}$ by Lemma~\ref{Lem:Channel}.

We also require that operator $R_P$ in Eqn.~(\ref{Eqn:PinR}) is such that $P \in \cent(\Gch)$ such that $\Upsilon(p,g) = 0$. Otherwise, $R_P \not\in \cent(\mathcal{M})$ and is therefore not a logical operator of $\mathcal{F}$ by the definition of a subsystem code. 

Finally, elements $R \in \Gch $ are elements of $\Gout$ as we showed above, and are thus not logical operators. We therefore see that
\begin{equation}
P \in \cent(\Gch) \backslash \Gch. \label{Eqn:LoutPT2}
\end{equation}
Combining Eqns.~(\ref{Eqn:LoutPT1}) and~(\ref{Eqn:LoutPT2}) verify Eqn.~(\ref{Eqn:OutputLogical}).
\end{proof1b}

We also have that $P \in \left(\cent(\Gin) \backslash \Gin\right) \cap \Gch $   are elements of $\mathcal{M}$, and are therefore measured under the foliation process. We see this by considering $R'_P(t)$ as in Eqn.~(\ref{Eqn:PinRt}). For elements $P \in \Gch $ there exists a $\xi(P)$ such that $R'_P(t) C[\xi(P)(t)] \in \mathcal{M}$. 
\end{proof1}

We require a representative operator $R_Q \in \left( \cent(\mathcal{R}) \backslash \mathcal{R} \right) \cap \cent(\mathcal{M}) $ of the form of Eqn.~(\ref{Eqn:PinR}) to propagate the logical information to the output state. It is worthwhile writing this explicitly as its value needs to be inferred from $\mathcal{M}$ at the point of readout. In some cases, it may be possible to choose
\begin{eqnarray}
R_Q &=& \prod_{j \in |q^X|} \Sigma_j^X(D) \sigmax{Z_j(D+1)} \\ 
&& \prod_{j \in |q^Z|} \Sigma_j^Z(D+1) \sigmaz{Z_j(D+1)}, \nonumber
 \end{eqnarray}
where, by definition, $\Upsilon(q,g) = 0$ for all $G \in \Gres$ with $(g^X \, g^Z)^T =  v(G)$. Since $\Sigma_j^X(D),\, \Sigma_j^Z(D+1) \in \mathcal{M} $ we have 
\begin{equation}
R_Q \sim \prod_{j \in |q^X|} \sigmax{Z_j(D+1)} \prod_{j \in |q^Z|} \sigmaz{Z_j(D+1)} ,
\end{equation} 
supported on the output system. Sometimes, however, this operator is not suitable because, perhaps, some of the qubits that support $R_Q$ are not available due to loss, or because we require the evaluation of an alternative representative at the output system for later information processing. In which case, we are free to multiply $R_Q$ by stabilizer operators to change its support.

\begin{proof2}
We next verify the elements of the stabilizer group. We consider the term $R'_P(t) $ and $R'_P(t-1)$ in Eqn.~(\ref{Eqn:PinRt}) where $P \in \Gres$ such that $\Upsilon(p,g) = 0$. We take the product of the two terms to give
\begin{eqnarray}
R'_p(t)R'_P(t  -  1) &=& \prod_{G \in \Gres}  \left( \sigmaz{G(t  - 1)} \sigmaz{G(t)}  \right)^{p^X \cdot g^Z} \nonumber  \\
&& \times  \prod_{j \in \supp{p^X}} \Sigma^X_j(t  - 1) \Sigma^X_j(t) \label{Eqn:RdashRdashdash} \\
&& \times  \prod_{j \in \supp{p^Z}}\Sigma^Z_j(t  - 1) \Sigma^Z_j(t) \nonumber \\
&& \times  \prod_{j \in \supp{p^X}}  \sigmaz{X_j(t  - 1)} \sigmaz{X_j(t)} \nonumber \\
&& \times  \prod_{j \in \supp{p^Z}}\sigmaz{Z_j(t  -  1)} \sigmaz{Z_j(t)}, \nonumber 
\end{eqnarray}
where $R'_P(t) R'_P(t-1) \in \mathcal{R}$. The product of this term with $C[P(t)] C[P(t-1)] \in \mathcal{R}$ where $C[P(t)]$ is defined in Eqn.~(\ref{Eqn:ClusterTerm}) gives $S_{\text{bulk}}[G(t)] \in \mathcal{R} \cap \mathcal{M}$ of Eqn.~(\ref{Eqn:StabilizerCell}).

We finally show that $S_{\text{bdry.}}[G(t)]$ of Eqn.~(\ref{Eqn:BdryStabilizer}) belongs to $\mathcal{S}$. This is shown by considering again Eqn.~(\ref{Eqn:PinR}) where $P \in \Gin$ such that $R'_P(t) \in \mathcal{R}$. Then, taking the product of $R'_P(t)$ and $C[\xi(P)(t)] \in \mathcal{R}$, as defined in Eqn.~(\ref{Eqn:Cap}), gives the desired operator which is included in $\mathcal{M}$.
\end{proof2}

\section{Foliating subsystem codes}
\label{App:SubsystemFoliation}

In the main text we focused on the foliation of stablizer codes. In fact, we find that our method for foliation extends to certain classes of subsystem codes as well with minor modifications to the scheme we have given above.

We consider a foliated channel where $\Gres$ is a non-Abelian generating set for subsystem code $\Gch$. The input code $\Gin$ may also be a subsystem code. However, we will not be interested in its gauge degrees of freedom, only its logical operators and its stabilizers. As such, without loss of generality, we will continue to denote the stabilizers of this system as $\Gin$ as before to maintain consistency with the theorems given above. The gauge and logical degrees of freedom can both be regarded as logical operators. This simplification allows us to keep our definition of the channel system, Def.~\ref{Def:Channel}, unchanged in our generalization to subsystem codes.

Further, we keep the ancilla system and the measurement pattern the same following the prescription set by $\Gres$. We index elements of the ancilla system with labels $G(t)$ where $G$ denotes an element of the non-Abelian generating set $\Gres$ and $t$ denotes a time interval. We then define elements of the stablizer group of the ancilla system such that we have $ \sigmax{G(t)} \in \mathcal{A}$ for all $t$ and gauge generators $G \in \Gres$. Likewise, we keep our definition of the measurement pattern, Def.~\ref{Def:Measure}, where we use only type-I foliated qubits in the channel system\footnote{Though we remark that generalizing to make use of type-II foliated qubits is straight forward.}. Explicitly, we measure all the qubits of the channel system in the Pauli-X basis. Ancilla qubits are measured in the basis $\sigmax{G(t)} \sigmaz{G(t)}^{g^X \cdot g^Z} \in \mathcal{M}$ where $(g^X \, g^Z)^T = v(G)$ for all $G\in\Gres$.

We modify the definition of the resource state. We modify Def.~\ref{Def:Resource} such that $V = 1$. We give our motivation for this choice shortly. We find that this modification is suitable for a large class of subsystem codes which, among others, includes CSS subsystem codes such as the Bacon-Shor code~\cite{Bacon06}, the subsystem surface code~\cite{Bravyi13}, the gauge color code~\cite{Bombin15} and variations of these models~\cite{Bacon06a, Bravyi11, Bravyi15, OConner16, Jones16}. Written explicitly, for a given $\Gres$ we have	
\begin{equation}
	\mathcal{R} = U^A  \left(  \chan \otimes \mathcal{A} \right), 
	\end{equation}
	where the ancilla system is in the product state
	 \begin{equation}
	\mathcal{A}  = \left\{  \sigmax{G(t)} : \forall t,\, G \in \mathcal{G}_\mathcal{R} \right\},
	 \end{equation}
and the entangling unitary $U^A$ is given by
	\begin{equation}
	U^A =  V \prod_{G\in\Gres, \, t}  U[G(t)] , \label{Eqn:UAsubsystem}
	\end{equation} 
	where now $V = 1$ and
	\begin{equation}
	U[G(t)] = \prod_{j \in \supp{g^X}} U^Z[X_j(t), G(t)] \prod_{j \in \supp{g^Z}} U^Z[Z_j(t), G(t)],
	\end{equation}
with $(g^X \, g^Z)^T = v(G)$ for each $G\in \Gres$.

Using the definitions given above, we state some facts about elements of $\mathcal{F}$ and their inclusion in $\mathcal{R} $ and $\mathcal{M}$ without proof. Instead, we only remark that the following statements are proven using the methodology given above where $\Gch$ is replaced with a non-Abelian group. To approach this discussion, we consider the following operators for arbitrary Pauli operators $P \in \cent(\Gch)$ with $(p^X \, p^Z)^T = v(P)$. We consider
\begin{eqnarray}
R'_P(t) &=& \prod_{j\in |p^X|} \Sigma^X_j(t)\sigmaz{X_j(t)} \prod_{j\in |p^Z|} \Sigma^Z_j(t) \sigmaz{Z_j(t)} \nonumber \\
&& \times  \prod_{G \in \Gres} \sigmaz{G(t)}^{p^X \cdot g^Z}  \prod_{ t' < t} \sigmaz{G(t')}^{\Upsilon (p ,g)} \nonumber \\
\end{eqnarray}
and
\begin{eqnarray}
R_P &=& \prod_{j\in |p^X|} \Sigma^X_j(D)\sigmax{Z_j(D+1)} \nonumber \\ 
&& \times \prod_{j\in |p^Z|} \Sigma^Z_j(D+1) \sigmaz{Z_j(D+1)} \nonumber \\
&& \times \prod_{\substack{G \in \Gres, \\ t\not=D+1}}  \sigmaz{G(t)}^{\Upsilon (p ,g)}  \label{Eqn:FullTube}
\end{eqnarray}
While we have written the term explicitly here, for $P \in \cent  (\Gch)$ we have $\Upsilon(p,g) = 0$ for all $g = v(G)$ and $G \in \Gres$ by the definition of a subsystem code. We therefore neglect terms with exponents of $\Upsilon(p,q)$ from the above two equations hereon.
We also consider the operator
\begin{eqnarray}
C[\xi(P)(t)] &=& \prod_{G \in \xi(P)} C[G(t)] \\
&=& \prod_{G \in \xi(P)} \sigmax{G(t)} \sigmaz{G(t)}^{g^X\cdot g^Z} \nonumber \\
&& \times \prod_{j\in |p^X|} \sigmaz{X_j(t)} \prod_{j\in |p^Z|} \sigmaz{Z_j(t)}  \nonumber
\end{eqnarray}
where we remember that $\xi(P) $ denotes a subset of elements of $\Gres$ whose product gives $P$ as defined in the main text. A set $\xi(P)$ must exist for elements $P \in \Gch$ by definition.

For $P \in \mathcal{C}(\mathcal{G}_{\text{ch.}} ) \cap \mathcal{G}_{\text{ch.}} \cap \mathcal{G}_{\text{in}} $ the product of these two terms gives us a stabilizer element of $\mathcal{R}$
\begin{eqnarray}
C[\xi(P)(t)] R'_P(t) &=& \prod_{\tilde{G} \in \xi(P)} \sigmax{\tilde{G}(t)} \sigmaz{\tilde{G}(t)}^{\tilde{g}^X\cdot \tilde{g}^Z} \nonumber \\
 && \times  \prod_{j\in |p^X|} \Sigma^X_j(t) \prod_{j\in |p^Z|} \Sigma^Z_j(t) \label{Eqn:SubsystemStabilizer}  \\
&& 
\times \prod_{G \in \Gres} \sigmaz{G(t)}^{p^X \cdot g^Z} . \nonumber
\end{eqnarray}
Likewise, for $P \in \cent(\Gch) \cap \Gch \backslash \Gin $ we have stabilizer generators
\begin{equation}
C[\xi(P)(t-1)]  R'_P(t-1) C[\xi(P)(t)]   R'_P(t) \in \mathcal{R}.
\end{equation}

\begin{figure}
\includegraphics{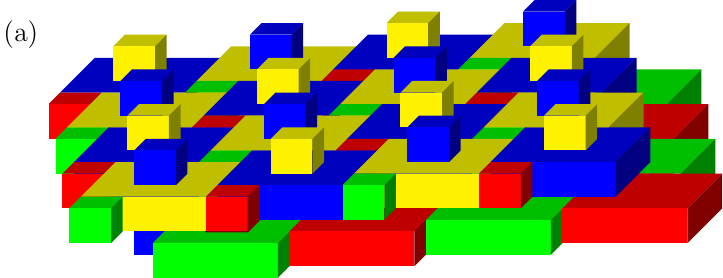}
\includegraphics[scale=0.6]{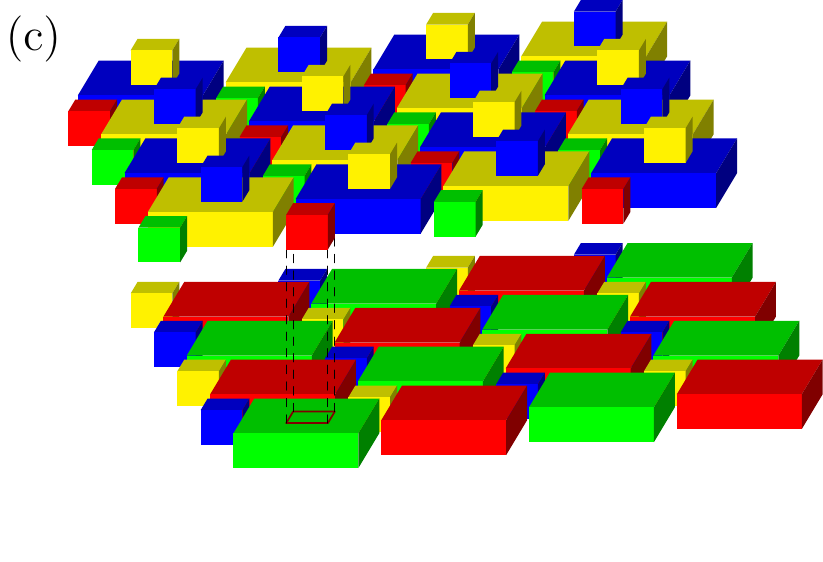}\includegraphics[scale=0.5]{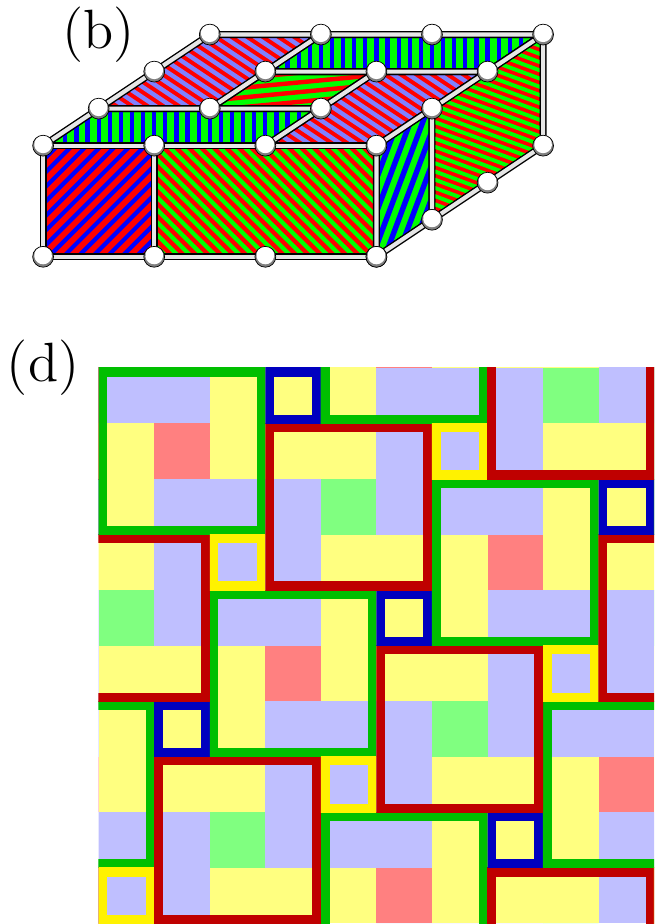}
\caption{\label{Fig:GCC} The gauge color code lattice arranged with the qubits lying on the vertices of a cubic lattice with four-colorable cells. (a)~The stabilizers are eight and thirty-two body terms living on the cuboidal cells of the lattice. (b)~The faces of a yellow cell whose faces are colored with a pair of colors opposite to the colors of the two cells that share the face. (c)~The cells are separated to reveal the structure of the lattice more clearly. The gauge terms lie on faces where pairs of cells share common support. (d)~A two dimensional representation of how cells lying on the layer above, marked with bold outlines, lie atop the layer below, where the cells are filled with pale colors.}
\end{figure}

In Eqn.~(\ref{Eqn:SubsystemStabilizer}) lies the issue with the foliation of general subsystem codes. In particular, it is not clear in general if $C[\xi(P)(t)] R'_P(t)  $ lies in $ \mathcal{M}$ due to the term 
 on the third line of Eqn.~(\ref{Eqn:SubsystemStabilizer})
 \begin{equation}
\prod_{G \in \Gres} \sigmaz{G(t)}^{p^X \cdot g^Z}  = \sigmaz{G(t)}^{p^X \cdot  \sum_{G \in \Gres}g^Z} .
 \end{equation}
 Since $\sigmaz{G(t)} \not\in \mathcal{M}$ for any $G$ or $t$, we rely on $p^X \cdot  \sum_{G \in \Gres}g^Z = 0$ for all $P \in \cent(\Gres) \cap \Gres$ and $G \in \Gres$.

This issue was easily dealt with in the case of stabilizer codes with inclusion of the $V$ operator. Indeed, in the case where $\Gres$ is Abelian we had that $\xi(P) =  \left\{P \right\}$ for all stabilizers $P \in \Gres $, and $P$ and $Q \in \mathcal{\Gres}$ commuted such that $p^X \cdot q^Z = q^X \cdot p^Z$. We were therefore able to eliminate all the residual Pauli-Z terms in the stabilizer group elements analogous to those in Eqn.~(\ref{Eqn:SubsystemStabilizer}) by simply coupling the ancillas with the controlled-phase gate $U^Z[P(t), Q(t)]$ for all $P,\,Q \in \Gres$ such that $p^X \cdot q^Z  = 1$. However, in the case of subsystem codes, where we use multiple ancilla qubits to infer the value of a stabilizer, it is not clear which ancilla qubits we should couple to nullify the extra Pauli-Z terms. We leave the general solution to this problem to future work.

Provided the Pauli-Z terms of the resource state can be cancelled out we can use the operators in Eqns.~(\ref{Eqn:FullTube}) and~(\ref{Eqn:SubsystemStabilizer}) to infer the stabilizers of the foliated system, as in Theorem~\ref{Thm:ChannelStabilizers}, and we can determine the output stabilizer group as in Theorem.~\ref{Thm:OutputCode}. We finally remark that the logical degrees of freedom, $P \in ( \mathcal{C}(\mathcal{G}_{\text{in}}) \cap \mathcal{C}(\mathcal{G}_{\text{ch.}}) ) \backslash ( \mathcal{G}_{\text{in}} \cup \mathcal{G}_{\text{ch.}})$, propagate through the foliated channel, and elements $P \in ( \mathcal{C}(\mathcal{G}_{\text{in}}) \backslash \mathcal{G}_{\text{in}} ) \cap \mathcal{G}_{\text{ch.}}$ can be measured by the channel using the operator shown in Eqn.~(\ref{Eqn:SubsystemStabilizer}), again, provided $p^X \cdot  \sum_{G \in \Gres}g^Z = 0$. This allows us to generalize Theorem~\ref{Thm:OutputCode} for subsystem codes where $p^X \cdot  \sum_{G \in \Gres}g^Z = 0$ for all $P$.

There are many subsystem codes that can be foliated without choosing a nontrivial $V$ operator. The CSS codes are natural candidates since 
\begin{equation}
 \sigmaz{G(t)}^{p^X \cdot  \sum_{G \in \Gres}g^Z}  =  \sigmaz{G(t)}^{p \cdot  \sum_{G \in \Gres}g} = 1, 
\end{equation}
by the definition of a CSS code. We therefore find that $V=1$ is suitable to learn stabilizer data. We also find that the subsystem color code~\cite{Bombin10a} can be foliated with our prescription.

\begin{figure}[t]
\includegraphics{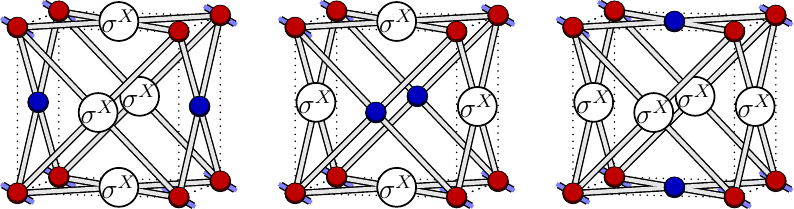}
\caption{The stabilizers of a single cell of the resource state for the gauge color code on red qubits indexed $X(t)$. Measuring the face terms of the gauge color code infers the value of the stabilizer three times. Taking the product of pairs of these cell measurements gives additional stabilizer data for the foliated system. \label{Fig:SingleShot}}
\end{figure}

It is interesting to study foliated subsystem codes as they give us another perspective on fault-tolerant quantum error correction. As an example, we consider single-shot error correction with the gauge color code~\cite{Bombin15, Bombin15a, Brown16a}. We give an alternative representation of the lattice of the gauge color code in Fig.~\ref{Fig:GCC} where the qubits lie on the vertices of a cubic lattice.

The gauge group of the gauge color code has elements $B^X_f  = \prod_{j\in \partial f}X_j $ and $B^Z_f  = \prod_{j\in \partial f}Z_j $ for all the faces of the lattice $f$, where $\partial f$ are the set of qubits that touch the face, and the product of a subset of the faces living on the boundary of a cell give a cell stabilizer. Specifically, the faces surrounding the cube are three-colored, see Fig.~\ref{Fig:GCC}(b). The product of the face terms of all of the faces of one particular color of a cell gives the value of a stabilizer for the corresponding cell. However, as we measure all of the faces, we redundantly learn the value of cell stabilizers three times. This redundancy enables us to predict the locations of measurement errors more reliably as each stabilizer is constrained to give the same value.

We now briefly look at these constraints from the perspective of foliation. In Fig.~\ref{Fig:SingleShot} we show a single cubic cell of the gauge color code on the red qubits indexed $X(t)$. It is readily checked that the operators shown in the figure that are the product of the face operators on two of the different colors of each cell are stabilizers of the foliated system $\mathcal{F} = \mathcal{R} \cup \mathcal{M}$. These stabilizers are unlike the stabilizers discussed in the main text where, if we exclude boundary stabilizers, we require measurements from qubits in different time intervals to learn the value of a stabilizer. We therefore find we have additional stabilizer data for error correction with foliated single-shot codes within each time interval. Indeed, a natural extension of our model of quantum computation is to perform gauge fixing between the three-dimensional color code and the gauge color code~\cite{Paetznick13, Bombin15} via foliated channels. Moreover, unlike many of the examples we have considered where we have assumed a large number of time intervals, due to its single-shot nature~\cite{Bombin15a} a fault-tolerant channel can be achieved with a small constant number of time intervals.

\section{Compressed foliation}\label{App:Compressed}
In this Appendix we discuss compressed foliation. This method includes additional check measurements of elements of $\mathcal{G}_{\mathcal{R}}$ in additon to those defined in $\mathcal{F}$ in the main text. The construction is similar to that described in Sec.~\ref{Sec:CodeFoliation}, but where additional ancillae and entangling unitaries are added, as we now describe.

In the construction of the new resource state, $\mathcal{R}$, one begins with a channel system $\chan$ according to Def.~1. The ancilla system $\mathcal{A}$ is constructed by using two ancillae for each element of $\mathcal{G}_{\mathcal{R}}$ and $t$, giving
\begin{equation}
\mathcal{A} = \left\{ \sigma^X[G(t)] , \sigma^X[G^C(t)] : \forall t, G\in \mathcal{G}_{\mathcal{R}} \right\},
\end{equation} 
where $G(t)$ and $G^C(t)$ label the coordinates of ancillae and the superscript $C$ denotes the additional ancilla for each $G$ and $t$ for the compressed system.

One entangles the ancillae to the channel using the unitary 
\begin{equation}
U^A = V  \prod_{G\in\Gres, \, t}  U[G(t)] U[G^C(t)],
\end{equation} 
where $U[G(t)]$ is defined in the main text and 
\begin{align}
U[G^C(t)] = &\prod_{j \in \supp{g^X}} U^Z[X_j(t), G^C(t)] \nonumber \\
\quad &\prod_{j \in \supp{g^Z}} U^Z[Z_j(t+1), G^C(t)], 
\end{align}
for $(g^X \, g^Z)^T = v(G)$.
Moreover, we update the operator $V$ with $V = \prod_t V(t) V^C(t)$ such that
\begin{equation}
V^C(t) = \prod_{\langle\langle G, H \rangle \rangle } U^Z\left[G^C(t,a) ,H^C(t,a)\right]^{g^X\cdot h^Z}
\end{equation}
where again, the product is over all unordered pairs from $G\neq H \in \Gres$, with $(g^X\,g^Z)^T = v(G)$ and $(h^X\,h^Z)^T = v(H)$, and $V(t)$ is also defined in the main text. 

In the case of CSS codes, compressing the foliation does not lead to any novel channels. In essence, each $G\in \Gres$ is measured twice per time interval. In the non-CSS codes compressed foliation leads to qualitatively different channels. In particular, compressed foliation results in resource states with higher degree, but can result in lower weight stabilizers. With a local basis change, one can always take a CSS code to a non-CSS code~\cite{Wen03, Nussinov09, Brown11}, and foliation of the two can lead to drastically different resource states. One needs to assess which channel is more suitable for a given purpose.


\end{document}